\documentclass[11pt]{article}

\usepackage{amsthm}
\usepackage{graphicx} % support the \includegraphics command and options
\usepackage{array} % for better arrays (eg matrices) in maths

\usepackage{amsmath, amssymb, amsfonts, verbatim}
\usepackage{hyphenat,epsfig,subfigure,multirow}
\usepackage{nicefrac}
\usepackage[utf8]{inputenc}
\usepackage{paralist}

\usepackage[usenames,dvipsnames]{xcolor}
\usepackage[ruled]{algorithm2e}

\DeclareFontFamily{U}{mathx}{\hyphenchar\font45}
\DeclareFontShape{U}{mathx}{m}{n}{
      <5> <6> <7> <8> <9> <10>
      <10.95> <12> <14.4> <17.28> <20.74> <24.88>
      mathx10
      }{}
\DeclareSymbolFont{mathx}{U}{mathx}{m}{n}
\DeclareMathSymbol{\bigtimes}{1}{mathx}{"91}

\usepackage{tcolorbox}
\tcbuselibrary{skins,breakable}
\tcbset{enhanced jigsaw}

\usepackage[normalem]{ulem}
\usepackage[compact]{titlesec}

\definecolor{DarkRed}{rgb}{0.5,0.1,0.1}
\definecolor{DarkBlue}{rgb}{0.1,0.1,0.5}

\usepackage{nameref}
\definecolor{ForestGreen}{rgb}{0.1333,0.5451,0.1333}
\definecolor{Red}{rgb}{0.9,0,0}
\usepackage[linktocpage=true,
	pagebackref=true,colorlinks,
	linkcolor=DarkRed,citecolor=ForestGreen,
	bookmarks,bookmarksopen,bookmarksnumbered]
	{hyperref}
\usepackage[noabbrev,nameinlink]{cleveref}
\crefname{property}{property}{Property}
\creflabelformat{property}{(#1)#2#3}
\crefname{equation}{eq}{Eq}
\creflabelformat{equation}{(#1)#2#3}
\usepackage{bm}
\usepackage{url}
\usepackage{xspace}
\usepackage[mathscr]{euscript}

\usepackage{tikz}
\usetikzlibrary{arrows}
\usetikzlibrary{arrows.meta}
\usetikzlibrary{shapes}
\usetikzlibrary{backgrounds}
\usetikzlibrary{positioning}
\usetikzlibrary{decorations.markings}
\usetikzlibrary{patterns}
\usetikzlibrary{calc}
\usetikzlibrary{fit}
\usetikzlibrary{decorations}

\usepackage{mdframed}

\usepackage[noend]{algpseudocode}
\makeatletter
\def\BState{\State\hskip-\ALG@thistlm}
\makeatother

\usepackage{cite}
\usepackage{enumitem}

\usepackage[margin=1in]{geometry}

\newtheorem{theorem}{Theorem}
\newtheorem{lemma}{Lemma}[section]
\newtheorem{proposition}[lemma]{Proposition}
\newtheorem{corollary}[lemma]{Corollary}
\newtheorem{claim}[lemma]{Claim}
\newtheorem{fact}[lemma]{Fact}

\newtheorem{definition}[lemma]{Definition}

\newtheorem*{claim*}{Claim}
\newtheorem*{assumption*}{Assumption}
\newtheorem*{proposition*}{Proposition}
\newtheorem*{lemma*}{Lemma}
\newtheorem*{problem5*}{Problem}

\crefname{lemma}{Lemma}{Lemmas}
\crefname{claim}{claim}{claims}

\newtheorem{mdresult}{Result}

\newtheorem{conjecture}[lemma]{Conjecture}
\newtheorem{observation}[lemma]{Observation}

\newtheoremstyle{restate}{}{}{\itshape}{}{\bfseries}{~(restated).}{.5em}{\thmnote{#3}}
\theoremstyle{restate}
\newtheorem*{restate}{}

\allowdisplaybreaks

\renewcommand{\qed}{\nobreak \ifvmode \relax \else
      \ifdim\lastskip<1.5em \hskip-\lastskip
      \hskip1.5em plus0em minus0.5em \fi \nobreak
      \vrule height0.75em width0.5em depth0.25em\fi}

\newcommand{\Ot}{\ensuremath{\widetilde{O}}}
\newcommand{\eps}{\ensuremath{\varepsilon}}

\newcommand{\Bracket}[1]{\Big[#1\Big]}

\newcommand{\paren}[1]{\ensuremath{\left(#1\right)}\xspace}
\newcommand{\card}[1]{\left\vert{#1}\right\vert}

\newcommand{\IR}{\ensuremath{\mathbb{R}}}

\newcommand{\IN}{\ensuremath{\mathbb{N}}}

\newcommand{\set}[1]{\ensuremath{\left\{ #1 \right\}}}

\DeclareMathOperator*{\Exp}{\ensuremath{{\mathbb{E}}}}
\DeclareMathOperator*{\Prob}{\ensuremath{\textnormal{Pr}}}
\renewcommand{\Pr}{\Prob}

\newenvironment{tbox}{\begin{tcolorbox}[
		enlarge top by=5pt,
		enlarge bottom by=5pt,
		 breakable,
		 boxsep=0pt,
                  left=4pt,
                  right=4pt,
                  top=10pt,
                  arc=4pt,
                  boxrule=1pt,toprule=1pt,
                  colback=white,
                  ]%%
	}
{\end{tcolorbox}}

\newcommand{\rv}[1]{\ensuremath{{\mathsf{#1}}}\xspace}
\newcommand{\rA}{\rv{A}}
\newcommand{\rB}{\rv{B}}

\newcommand{\rW}{\rv{W}}
\newcommand{\rX}{\rv{X}}
\newcommand{\rY}{\rv{Y}}
\newcommand{\rZ}{\rv{Z}}
\newcommand{\rM}{\rv{M}}

\newcommand{\II}{\ensuremath{\mathbb{I}}}
\newcommand{\HH}{\ensuremath{\mathbb{H}}}

\newcommand{\mireal}[1][]{
  \ifx\relax#1\relax%
    \II(\mione \,; \mitwo)%
  \else%
    \II(\mione \,; \mitwo\mid #1)%
  \fi
}
\newcommand{\en}[1]{\ensuremath{\HH(#1)}}

\newcommand{\cX}{\ensuremath{\mathcal{X}}}
\newcommand{\cE}{\ensuremath{\mathcal{E}}}
\newcommand{\cA}{\ensuremath{\mathcal{A}}}

\newcommand{\player}[1]{\ensuremath{\mathcal{P}_{#1}}}
\newcommand{\chain}{\textnormal{\textsc{chain}}}
\newcommand{\augchain}{\textnormal{\textsc{aug-chain}}}
\newcommand{\indexprob}{\textnormal{\textsc{index}}}

\newcommand{\cD}{\mathcal{D}}

\newcommand{\prot}{\ensuremath{\pi}}
\newcommand{\Prot}{\ensuremath{\Pi}}

\newcommand{\ind}{\ensuremath{\sigma}}

\newcommand{\biasedind}{\ensuremath{\textnormal{\textsc{bias-ind}}}}

\newcommand{\rrho}{\ensuremath{\bm{\rho}}}
\newcommand{\cLhalf}{\ensuremath{\mathcal{L}}}

\title{Optimal Communication Complexity of Chained Index}
\author{Janani Sundaresan\footnote{(jsundaresan@uwaterloo.ca) Cheriton School of Computer Science, University of Waterloo. Supported in part by Sepehr Assadi's Sloan Research Fellowship and startup grant from University of Waterloo. \smallskip}}

\date{}

\begin{document}
\maketitle

%\pagenumbering{roman}

% !TeX root = main.tex 
%!TEX root = main.tex
\begin{abstract}
	We study the $\chain$ communication problem introduced by Cormode et al. [ICALP 2019]. 
	For $k\geq 1$, in the $\chain_{n,k}$ problem, there are $k$ string and index pairs $(X_i, \ind_i)$ for $i \in [k]$ such that the value at position $\ind_i$ in string $X_i$ is the same bit for all $k$ pairs. The input is shared between $k+1$ players as follows.
	Player 1 has the first string $X_1 \in \{0,1\}^n$, player 2 has the first index $\ind_1 \in [n]$ and the second string $X_2 \in \{0,1\}^n$, player 3 has the second index $\ind_2 \in [n]$ along with the third string $X_3 \in \{0,1\}^n$, and so on. 
	Player $k+1$ has the last index $\ind_k \in [n]$. 
	The communication is one way from each player to the next, starting from player 1 to player 2, then from player 2 to player 3 and so on. Player $k+1$, after receiving the message from player $k$, has to output a single bit which is the value at position $\ind_i$ in $X_i$ for any $i \in [k]$. It is a generalization of the well-studied $\indexprob$ problem, which is equivalent to $\chain_{n, 2}$.
	%There is a trivial protocol which uses $n$ bits, as player $k$ can send the entire string $X^k \in \{0,1\}^n$ to player $k+1$. 
	
		Cormode et al. proved that the $\chain_{n,k}$ problem requires $\Omega(n/k^2)$ communication, and they used it to prove streaming lower bounds for the approximation of maximum independent sets. Subsequently, Feldman et al. [STOC 2020] used it to prove lower bounds for streaming submodular maximization. However, it is not known whether the $\Omega(n/k^2)$ lower bound used in these works is optimal for the problem, and in fact, it was conjectured by Cormode et al. that $\Omega(n)$ bits are necessary.  
	
	We prove the optimal lower bound of $\Omega(n)$ for $\chain_{n,k}$ when $k = o(n/\log n)$  as our main result. This settles the open conjecture of Cormode et al., barring the range of $k = \Omega(n /\log n)$.  The main technique is a reduction to a non-standard $\indexprob$ problem where the input to the players is such that the answer is biased away from uniform. This biased version of $\indexprob$ is analyzed using tools from information theory. As a corollary, we get an improved lower bound for approximation of maximum independent set in vertex arrival streams via a reduction from $\chain$ directly.
\end{abstract}

\clearpage

\setcounter{tocdepth}{3}
\tableofcontents

\clearpage

\section{Introduction}

The $\indexprob$ problem is one of the foundational problems in communication complexity. For $n \geq 1$, in the $\indexprob_n$ problem, there are two players Alice and Bob. Alice has a string $X \in \{0,1\}^n$ and Bob has an index $\ind \in [n]$, and Bob has to output the value of $X$ at position $\ind$. If the communication is one-way from Alice to Bob, it is easy to show that Alice needs to send $\Omega(n)$ bits to get any constant advantage  \cite{Ablayev96,KremerNR95}. 
This problem has been well-studied in multiple settings, and we know tight trade-offs in the two-party communication model for communication complexity \cite{MiltersenNSW98}, information complexity \cite{JainRS09}, and quantum communication complexity \cite{BuhrmanWolf01,JainRS09}.

Among the numerous applications of communication complexity, one that is of interest to us is proving lower bounds for streaming algorithms. $\indexprob$ and its variants, in particular, have been quite useful in this context, for example, in \cite{IndykW03, FeigenbaumKMSZ08, GuhaM09, GuruswamiO13, DarkK20, ChenKPSSY21}. This is by no means an exhaustive list. 

In this paper, we study a natural generalization of $\indexprob$, called chained index ($\chain_{n,k}$ for $n,k \geq 1$) introduced by \cite{CormodeDK19}. There are $k$ different instances of $\indexprob_n$, correlated so that they have the same answer. They are ``chained" together, where each player holds the index to the previous instance, and also the string for the next instance.  

\begin{definition}[Informal]\label{def:informal-chain}
	In $\chain_{n,k}$, there are $k$ instances of $\indexprob_n$, all with the same answer. 
	Players 1 and 2 take on the role of Alice and Bob respectively in the first instance, players 2 and 3 take on the role of Alice and Bob respectively for the second instance, and so on, all the way till players $k$ and $k+1$ for the last instance.
	
	Communication is one-way from each player to the next in ascending order. The last player has to output the answer.  The communication cost is the total number of bits in all the messages sent by the players. See \Cref{fig:informal-chain} for an illustration. 
\end{definition}

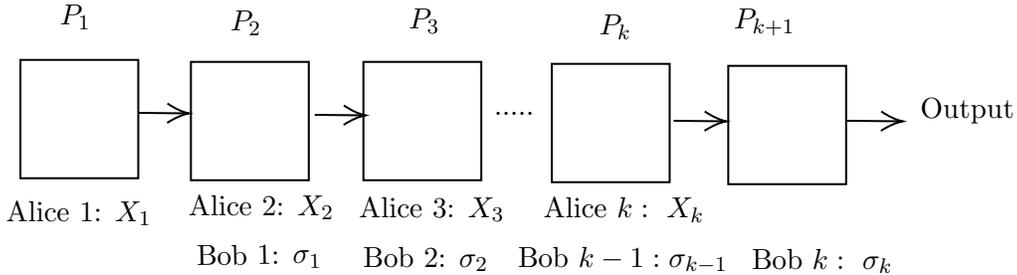
\begin{figure}[h!]
	\centering
	\tikzset{every picture/.style={line width=0.75pt}} %set default line width to 0.75pt        

\begin{tikzpicture}[x=0.75pt,y=0.75pt,yscale=-1,xscale=1]
	%uncomment if require: \path (0,409); %set diagram left start at 0, and has height of 409
	
	%Shape: Rectangle [id:dp49468729856407423] 
	\draw   (111.55,70.27) -- (171,70.27) -- (171,130) -- (111.55,130) -- cycle ;
	%Straight Lines [id:da17478687373564905] 
	\draw    (171,97) -- (194.55,97.25) ;
	\draw [shift={(196.55,97.27)}, rotate = 180.61] [color={rgb, 255:red, 0; green, 0; blue, 0 }  ][line width=0.75]    (10.93,-4.9) .. controls (6.95,-2.3) and (3.31,-0.67) .. (0,0) .. controls (3.31,0.67) and (6.95,2.3) .. (10.93,4.9)   ;
	%Straight Lines [id:da17476620658588926] 
	\draw    (260,98) -- (283.55,98.25) ;
	\draw [shift={(285.55,98.27)}, rotate = 180.61] [color={rgb, 255:red, 0; green, 0; blue, 0 }  ][line width=0.75]    (10.93,-4.9) .. controls (6.95,-2.3) and (3.31,-0.67) .. (0,0) .. controls (3.31,0.67) and (6.95,2.3) .. (10.93,4.9)   ;
	%Straight Lines [id:da01509882718820088] 
	\draw    (441,101) -- (465.55,101.25) ;
	\draw [shift={(467.55,101.27)}, rotate = 180.59] [color={rgb, 255:red, 0; green, 0; blue, 0 }  ][line width=0.75]    (10.93,-4.9) .. controls (6.95,-2.3) and (3.31,-0.67) .. (0,0) .. controls (3.31,0.67) and (6.95,2.3) .. (10.93,4.9)   ;
	%Straight Lines [id:da6556538511116647] 
	\draw    (528,101) -- (554.55,101.25) ;
	\draw [shift={(556.55,101.27)}, rotate = 180.55] [color={rgb, 255:red, 0; green, 0; blue, 0 }  ][line width=0.75]    (10.93,-4.9) .. controls (6.95,-2.3) and (3.31,-0.67) .. (0,0) .. controls (3.31,0.67) and (6.95,2.3) .. (10.93,4.9)   ;
	%Shape: Rectangle [id:dp3285073823844342] 
	\draw   (197.55,71.27) -- (257,71.27) -- (257,131) -- (197.55,131) -- cycle ;
	%Shape: Rectangle [id:dp14579349413484155] 
	\draw   (284.55,71.27) -- (344,71.27) -- (344,131) -- (284.55,131) -- cycle ;
	%Shape: Rectangle [id:dp17680960806573354] 
	\draw   (379.55,72.27) -- (439,72.27) -- (439,132) -- (379.55,132) -- cycle ;
	%Shape: Rectangle [id:dp1926226234606605] 
	\draw   (468.55,73.27) -- (528,73.27) -- (528,133) -- (468.55,133) -- cycle ;
	
	% Text Node
	\draw (349,95) node [anchor=north west][inner sep=0.75pt]   [align=left] {$\displaystyle .....$};
	% Text Node
	\draw (130,41) node [anchor=north west][inner sep=0.75pt]   [align=left] {$\displaystyle P_{1}$};
	% Text Node
	\draw (216,44) node [anchor=north west][inner sep=0.75pt]   [align=left] {$\displaystyle P_{2}$};
	% Text Node
	\draw (306,43) node [anchor=north west][inner sep=0.75pt]   [align=left] {$\displaystyle P_{3}$};
	% Text Node
	\draw (470,43) node [anchor=north west][inner sep=0.75pt]   [align=left] {$\displaystyle P_{k+1}$};
	% Text Node
	\draw (402,46) node [anchor=north west][inner sep=0.75pt]   [align=left] {$\displaystyle P_{k}$};
	% Text Node
	\draw (103,140) node [anchor=north west][inner sep=0.75pt]   [align=left] {Alice 1: $\displaystyle X_{1}$};
	% Text Node
	\draw (198.83,162.43) node [anchor=north west][inner sep=0.75pt]  [rotate=-358.79] [align=left] {Bob 1: $\displaystyle \sigma _{1}$};
	% Text Node
	\draw (283,163) node [anchor=north west][inner sep=0.75pt]   [align=left] {Bob 2: $\displaystyle \sigma _{2}$};
	% Text Node
	\draw (195,137) node [anchor=north west][inner sep=0.75pt]   [align=left] {Alice 2: $\displaystyle X_{2}$};
	% Text Node
	\draw (281,138) node [anchor=north west][inner sep=0.75pt]   [align=left] {Alice 3: $\displaystyle X_{3}$};
	% Text Node
	\draw (374,138) node [anchor=north west][inner sep=0.75pt]   [align=left] {Alice $\displaystyle k:\ X_{k}$};
	% Text Node
	\draw (361,163) node [anchor=north west][inner sep=0.75pt]   [align=left] {Bob $\displaystyle k-1:\sigma _{k-1}$};
	% Text Node
	\draw (479,164) node [anchor=north west][inner sep=0.75pt]   [align=left] {Bob $\displaystyle k:\ \sigma _{k}$};
	% Text Node
	\draw (564,88) node [anchor=north west][inner sep=0.75pt]   [align=left] {Output};

\end{tikzpicture}\caption{An illustration of the $\chain_{n,k}$ problem with $k$ correlated sub-instances of $\indexprob_n$ from \Cref{def:informal-chain}. The arrows illustrate that the message is from $P_i $ to $P_{i+1}$ for $i \in [k]$.}\label{fig:informal-chain}
\end{figure}

In \cite{CormodeDK19}, a reduction from $\chain$ was employed to get a lower bound for approximation of maximum independent sets in vertex arrival streams. 
Before its introduction, \cite{Kapralov12} used the problem implicitly to get a lower bound of $(1-1/e)$ in the approximation factor for maximum matching in $\tilde{O}(n)$ space in vertex arrival streams. 

The problem has been used by the breakthrough result of \cite{FeldmanNSZ20} to study the multi-party communication complexity of submodular maximization. They proved that any randomized $p$-party protocol which maximizes a monotone submodular function $f:\{0,1\}^N \rightarrow \IR$, subject to a cardinality constraint of at most $p$ and an approximation factor of at least $(1/2+\epsilon)$, uses $\Omega(N\epsilon/p^3)$ communication. This also gave a lower bound for streaming submodular maximization.  The $\chain$ problem was used by \cite{FeldmanNSZ22} also for similar purposes, but subject to stronger matroid constraints. \cite{BhoreKO22} used a reduction from $\chain$ to prove lower bounds for interval independent set selection in streams of split intervals. 

We do not know tight bounds for the communication complexity of the $\chain$ problem, despite finding varied applications of it. There is a trivial protocol of $O(n)$ bits, where any player can send the entire string to the next player who holds the index. 
Another simple protocol is for each player to send $O(n/k)$ bits randomly sampled from their strings using public randomness, and with constant probability, in at least one of the $k$ instances, we send the special bit to the player holding the index. However, this still takes $\Omega(n)$ total bits of communication. 

In \cite{CormodeDK19}, they prove a lower bound of $\Omega(n/k^2) $ for any $k \geq 1$  through a reduction from conservative multi-party pointer jumping problem, introduced by \cite{DammJS96}. They state without proof that a stronger lower bound of $\Omega(n/k)$ can be obtained for a restricted range of $k \leq O((n/\log n)^{1/4})$. 
They posed the following conjecture on the optimal communication lower bound. 

\begin{conjecture}[\!\!\cite{CormodeDK19}]\label{conj:open}
	Any protocol that solves $\chain_{n,k}$ requires $\Omega(n)$ bits of communication. 
\end{conjecture}

\cite{FeldmanNSZ20} made some progress on \Cref{conj:open} by showing that among all the messages sent by the players, there is at least one message with $\Omega(n/k^2)$ bits for every $k \geq 1$. But the original conjecture is still open, and this is the focus of our work. 

\subsection{Our Results}

We settle \Cref{conj:open} almost fully by proving the optimal lower bound of $\Omega(n)$ barring the corner case of when $k$ is too large. As far as we know, this corner case is not a focus for existing reductions from $\chain_{n,k}$.

\begin{theorem}\label{thm:main}
	For any $n, k \geq 1$, any protocol for $\chain_{n,k}$ with probability of success at least 2/3, requires $\Omega(n - k \log n)$ total bits of communication.
\end{theorem}

Therefore, as long as $k = o(n / \log n)$, we get the optimal $\Omega(n)$ lower bound from \Cref{thm:main}. 

The proof of \Cref{thm:main} can be found in \Cref{sec:lb}. We prove the lower bound in the more general blackboard model of communication, instead of private messages between players (see \Cref{sec:prelim-cc-model} for details). The main idea is to analyze $\indexprob_n$ where Alice and Bob already have \emph{some prior advantage} in guessing the answer. We elaborate on our techniques in \Cref{sec:overview-techniques}.

As a direct corollary of \Cref{thm:main}, we get improvements in streaming lower bounds in \cite{CormodeDK19,FeldmanNSZ20} through reductions from $\chain$ immediately. In particular, we get that any algorithm which $\alpha$-approximates the size of a maximum independent set in vertex arrival streams requires $\Omega(n^2/\alpha^5- \log n)$ space, while the previous bound was $\Omega(n^2/\alpha^7)$ in \cite{CormodeDK19}.   We present the implications of our result in \Cref{sec:applications}. %For most applications, including the ones in \Cref{sec:applications}, the range of $k$ is substantially smaller than $n$, so the range of $k$ in our lower bound is sufficient.

A further generalization of $\chain$, called Augmented Chain was defined in \cite{DarkDK23}. Here, instances of Augmented Index are chained together instead.  Our lower bound of $\Omega(n-k\log n)$ can be extended to Augmented Chain also, and the details are covered in \Cref{subsec:aug-chain}. 

\subsection{Our Techniques}\label{sec:overview-techniques}

In this subsection, we give an overview of the challenges in proving the lower bound and a summary of our techniques. We start by going over the prior techniques.

\paragraph{Prior Techniques.} 
%\worry{Talk about pointer jumping technique too}
We will briefly talk about the technique used in 
\cite{FeldmanNSZ20} to prove that there is at least one player who sends $\Omega(n/k^2)$ bits for any $k \geq 1$. We will argue that these techniques can be extended to proving a lower bound of $\Omega(n/k)$ for the total number of bits, but not all the way to $\Omega(n)$ bits.

The first step in proving a lower bound for $\chain_{n,k}$ in \cite{FeldmanNSZ20} is a decorrelation step -- the $k$ instances of $\indexprob_{n}$ have the same answer, and they remove this correlation with a hybrid argument. These arguments have been used extensively in the literature (see e.g., \cite{KapralovKS15,AssadiKSY20,KapralovMTWZ22,AssadiS23}). Intuitively, any protocol that tries to solve $\chain_{n,k}$ may attempt to solve ``many" of the instances of $\indexprob_n$, albeit each with a ``small" advantage over 1/2, in the hope that the ``small" advantages may accrue to get a constant probability of success overall (taking advantage of the fact that all the instances of $\indexprob_{n}$ have the same answer). The hybrid argument is a way to reduce proving a lower bound on the overall problem, to proving a lower bound for $k$ different $\indexprob_{n}$ problems against these low advantages. 

Let us assume, for simplicity, that the protocol tries to get an advantage of $\Omega(1/k)$ in each instance of $\indexprob_n$, to get a constant total advantage. We can prove that any protocol that gets an advantage of $\Omega(1/k)$ for $\indexprob_{n}$ uses $\Omega(n/k^2)$ bits of communication using basic tools from information theory \cite{ChakrabartiCKM13} (this is quite standard, see e.g., \cite{Assadi20} for a direct proof), and this is known to be tight. Therefore, for $k$ instances, we get a lower bound of $\Omega(n/k)$ bits in total. Now, we will argue why this is not the optimal lower bound.

On one hand, for $\indexprob_n$, it is known that for any $\delta \in (0,1/2)$, there is a protocol that uses $O(n\delta^2)$ bits of communication to get a probability of success $1/2+\delta$. This means that $\indexprob_{n}$ can be solved with advantage $\Omega(1/k)$ in  $O(n/k^2)$ bits; in other words, each ``hybrid step" of the previous lower bound argument is optimal. So, then, why can we not get a good protocol for $\chain_{n,k}$ by running the protocol for $\indexprob_n$ with $\delta = 1/k$ on all $k$ instances? This protocol would have $O(n/k)$ bits of communication in total. The reason is that the $ 1/k$ small advantages do not add up as we would like them to. To illustrate this, we will briefly talk about the protocol that gets $\delta$ advantage in $O(n\delta^2)$ bits for $\indexprob_n$. This protocol is presented formally in \Cref{app:index-ub} for the sake of completeness.

\paragraph{Protocol for $\indexprob_n$.}
First, we will sketch a protocol for $\delta = 1/\sqrt{n}$ that uses $O(1)$ bits of communication. 
Let us imagine that the input $X$ is chosen uniformly at random from $\{0,1\}^n$ and the index $\ind$ is chosen uniformly at random from $[n]$. Then, Alice finds the majority bit from her string $X$ and sends it to Bob. Bob just outputs the bit sent by Alice. We know, from simple anti-concentration bounds on the binomial distribution, that the number of indices with the majority bit in $X$ is at least $n/2 + c\sqrt{n}$ with constant probability for some appropriate constant $c$. The protocol succeeds if Bob holds the index $\ind$ to a majority bit, and this happens with probability at least $\frac12 + \frac{c}{\sqrt{n}}$. 

The assumption that the input $X$ is chosen uniformly at random from $\{0,1\}^n $ can be removed using public randomness. Alice and Bob collectively sample a random string $A \in \{0,1\}^n$, and Alice changes her input to $X \oplus A$ so that each bit is 0 or 1 with equal probability. Similarly, the assumption that the index $\ind$ is chosen uniformly at random from $[n]$ can be removed by Alice and Bob sampling a random permutation of $[n]$. Alice permutes string $X$ according to this permutation, and Bob changes his input to the index that the permutation maps $\sigma$ to. This is termed as the self-reducibility property of $\indexprob$.

We can also extend the protocol to any $\delta$ by partitioning the string $X$ into $n\delta^2$ blocks at random using public randomness and sending the majority bit in each block. Bob knows the block that his input index $\ind$ belongs to as public randomness is used.

\paragraph{Challenges for $\chain_{n,k}$.}
If we use the protocol we described for $\chain_{n,k}$, we are left with $k$ bits from each instance of $\indexprob_n$, which may be the correct answer to the problem with probability $1/2 + \Theta(1/k)$, but, the variance of each of these bits is $1/4 - \Theta(1/k^2)$. We have a protocol for $\chain_{n,k}$, where, out of the $k$ instances of $\indexprob_{n}$, it finds the right answer for $\approx k/2 + \Theta(1)$ instances in expectation. However, the standard deviation of the number of right answers is $\approx \sqrt{k}/2$, which is enough to mask the $\Theta(1)$ improvement over $k/2$ we get in expectation. If we use a hybrid argument over the total variation distance, each of the smaller ``hybrid steps" is optimal, whereas the overall lower bound is not optimal. We cannot achieve a lower bound stronger than $\Omega(n/k)$. 

\paragraph{Our Solution.}

Instead of keeping track of progress in terms of advantage gained in guessing the answer, we directly keep track of the ``information" the message reveals about the answer. 
Formally, this translates to the change in entropy of the answer, after each successive message.
Initially, the players have no information about the answer, and it is uniform over $\{0,1\}$ (the entropy is 1). Any protocol with a large enough probability of success must reduce the entropy of the answer by a large factor (see Fano's inequality in \Cref{prop:fanos}).
We prove that after each message, the entropy is reduced only by an additive factor \emph{linear in the length of the message}, by a reduction to the $\indexprob$ problem. This will give us a lower bound on the total length of the messages.

For $\indexprob_n$, in any protocol with probability of success at least $\eps$, the entropy of the answer conditioned on the message is at most $H_2(\eps)$ where $H_2 (x) = x \log (1/x) + (1-x) \log (1/(1-x))$ is the binary entropy function. In the standard version where the initial entropy is 1, the reduction in entropy  is $1- H_2(\eps)$, and it is known that the protocol requires   $\Omega(n (1-H_2(\eps)))$ communication \cite{JainRS09}. 
This is not sufficient for our application due to the following reason: after the message of $\player{1}$, when $\player{2}$ and $\player{3}$ attempt to solve  $\indexprob_n$, they already have some prior advantage that the message of $\player{1}$ gives them. Therefore, we need to analyze $\indexprob_{n}$ when the answer is not uniform over $\{0,1\}$.

\paragraph{Biased Index.}
We define the biased index problem, parametrized by $\theta \in [-1/2, 1/2]$. 
 Alice and Bob receive input $Y \in \{0,1\}^n$ and $\rho \in [n]$ respectively, such that the value at position $\rho$ in $Y$ (denoted by $Y(\rho)$) is 1 with probability $1/2 + \theta$ and 0 otherwise. Alice sends a message $M$ to Bob and Bob has to output $Y(\rho)$.  The initial entropy of the answer is $H_2(1/2 + \theta)$. We prove that entropy of $Y(\rho)$ conditioned on $M$ is smaller by at most $O((\card{M}+\log n)/n)$ compared to the initial $H_2 (1/2 + \theta)$, which is \textbf{our main contribution} (see \Cref{lem:bias-ind}).

We can show that the entropy of input to Alice, i.e. the random string $Y$, in such a distribution is at least $\Omega(n \cdot H_2 (1/2 + \theta))$.
Hence, after a message of length $s$ from Alice, the entropy of string $Y$ reduces to $ \Omega(n \cdot H_2 (1/2 + \theta) - s)$. For a randomly chosen position in $Y$ after conditioning on the message, the entropy is at least $ \approx H_2(1/2 + \theta) - s/n$, which gives a lower bound on $s$. 

In the distribution given to Alice and Bob, however, $\rho$ is \emph{not} chosen uniformly at random, and in fact, is correlated with the distribution of $Y$.  Such versions of index where the distributions of Alice and Bob are correlated have been studied before (see Sparse Indexing in Appendix A of \cite{AssadiKL16} and Section 3.3 of \cite{Saglam19}). This correlation is the main issue in analyzing biased index with information theoretic tools. 

Adapted from the approach in Appendix A of \cite{AssadiKL16}, we restrict the randomness in $Y$ to a fixed set of indices of a carefully chosen size (based on $\theta$), and break this correlation. The loss in entropy of $Y$ is not significant enough to hinder us, and the restriction then allows us to use standard information theoretic tools to analyze biased index (see \Cref{subsec:bias-ind} for more details). 

 \paragraph{Independent and Concurrent Work.}
 
Independently and concurrently of this work, \cite{HuangMYZ24} made progress on \Cref{conj:open}. They showed a lower bound of $\Omega(n/k+ \sqrt{n})$ for oblivious protocols (where the length of the message sent by each player does not depend on the input), and a lower bound of $\Omega(n/k-k)
$ for general protocols. We show a lower bound of $\Omega(n - k \log n)$ for all protocols.\footnote{The authors of \cite{HuangMYZ24} pointed out an important flaw in the arguments of an earlier version of this work which was posted at around the same time as \cite{HuangMYZ24}. This flaw was subsequently fixed in the current version using a global change to the original argument, which now recovers optimal result for $k = o(n/\log n)$.} 

Quantitatively, our lower bounds are a factor of almost $k$ stronger, and are optimal for $k = o(n/\log n)$; moreover, our lower bound also holds for the Augmented Chain problem of \cite{DarkDK23}. In terms of techniques, however, the two works are entirely disjoint: their proof is based on a new method of analysis through min-entropy and we use information theoretic approaches. 

\section{Preliminaries}\label{sec:prelim}

In this section, we will present the required notation and definitions for our proof. 

\paragraph{Notation.}
For any tuple $A = (A_1, A_2, \ldots, A_m)$ of $m$ items, we use $A_{<i} $ to denote the tuple $(A_1, A_2, \ldots, A_{i-1})$ for all $i \in [m]$. We use sans-serif font to denote random variables. For any random variable $\rA$, we use $A \sim \rA$ to denote any $A$ sampled from the distribution of the random variable $\rA$. 

For any string $X \in \{0,1\}^n$, we use $X(\sigma)$ to denote the bit at position $\sigma$ in $X$ for $\sigma \in [n]$. We use $X(<\sigma)$ to denote the string of $\sigma-1$ bits preceding $X(\sigma)$ in $X$.  We use $X(S)$  for any $S \subseteq [n]$ to denote the bits at positions in set $S$. 

%Let $\cU_n$ be the uniform distribution over $\{0,1\}^n$, and let $\cU_1$ be the uniform distribution over $\{0,1\}$.

For any $x \in [0,1]$, we use $H_2: [0,1] \rightarrow [0,1]$ to denote the binary entropy function. 
\[
	H_2(x) = -x \log x - (1-x) \log (1-x).
\]

We need the following standard approximation of binomial coefficients (see  Lemma 7 of Chapter 10 in \cite{MacWilliamsS77}). 

\begin{fact}[c.f. \cite{MacWilliamsS77}]\label{fact:binom-lower-bound}
	For any $p \geq 1$ and any $q \in[p-1]$, we have,
	\[
		2^{p \cdot H_2(q/p)}\cdot \sqrt{\frac{n}{8\pi q(p-q)}} \leq \binom{p}{q} \leq 2^{p \cdot H_2(q/p)}\cdot \sqrt{\frac{n}{2\pi q(p-q)}}
	\]
\end{fact}

%\begin{definition}\label{def:tvd}
%	For any two distributions $\mu$ and $\nu$ defined over support $\cA$, we define the total variation distance as,
%	\[
%	\tvd{\mu}{\nu} = \frac12 \cdot \sum_{a \in \cA} \card{\mu(a) - \nu(a)}.
%	\]
%\end{definition}

%Total variation distance is related to the maximum probability of success attainable when predicting the outcome of a Bernoulli distribution.
%\begin{fact}\label{fact:tvd-sample}
%	For any Bernoulli distribution $\mu$, for any fixed bit $b \in \{0,1\}$, the probability that a bit sampled from $\mu$ equals $b$ is at most 
%	\[
%		\frac12 + \tvd{\mu}{u_1}.
%	\]
%\end{fact}
%\begin{proof}
%	Let $p$ be the probability of getting outcome as 1 in $\mu$, then $1-p$ is the probability of getting outcome as 0 in $\mu$.
%	Let $\rv{b}' \in \{0,1\}$ be sampled from $\mu$. For any fixed $b \in \{0,1\}$,
%	\begin{align*}
%		\Pr[\rv{b}'=b] \leq  \frac12 + \card{p-\frac12} \leq \frac12 + \tvd{\mu}{u_1}. \qedhere
%	\end{align*}
%\end{proof}

%\begin{definition}[KL-divergence]\label{def:kl}
%	For any two distributions $\mu, \nu$ over the set $\cA$, we define the Kullback-Leibner divergence or KL-Divergence between $\mu$ and $\nu$, denoted by $\kl{\mu}{\nu}$ as,
%	\[
%	\kl{\mu}{\nu} = \sum_{a \in \cA} \mu(a) \cdot \log\paren{\frac{\mu(a)}{\nu(a)}},
%	\]
%	with the convention that for all $a \in \cA$, $0 \log (0/\nu(a)) = 0$ and $\mu(a) \cdot \log(\mu(a)/0) = \infty$. 
%\end{definition}

\subsection{Communication Complexity Model}\label{sec:prelim-cc-model}
We use the standard number-in-hand multi-party model of communication. Only the basic definitions are given in this subsection. More details can be found in textbooks on communication complexity \cite{RaoY20,KushilevitzN97}.

For any $k \geq 1$, let $f$ be a function from $\cA_1 \times \cA_2 \times \ldots \times \cA_k$ to $\{0,1\}$.
There are $k$ players $\player{1}, \player{2}, \ldots, \player{k}$ where $\player{i}$ gets an input $a_i \in \cA_i $ for $i \in [k]$. 
There is a shared blackboard visible to all the players. The players have access to a shared tape of random bits, along with their own private randomness. 
In any protocol $\prot$ for $f$, the players send a message to the blackboard in increasing order ($\player{1}$ sends a message followed by $\player{2}$, and so on till $\player{k}$). The last player $\player{k}$, after all the messages are sent, outputs a single bit denoted by $ \prot(a_1, a_2, \ldots, a_k)$. 
Protocol $\prot$ is said to solve $f$ with probability of success at least $1-\delta$ if, for all $i \in [k]$, for any choice of $a_i \in \cA_i$, we have,
\[
	\Pr[\prot(a_1, a_2, \ldots, a_k) \neq f(a_1, a_2, \ldots, a_k)] \leq \delta.  
\]

The \textbf{communication cost} of a protocol $\prot$ is defined as the worst case \textbf{total communication} of all the players on the blackboard at the end of the protocol.

%We call a protocol \textbf{oblivious} if the length of the messages sent by each player does not depend on the input to the protocol. A protocol is said to be \textbf{non-oblivious} otherwise. 

%\begin{proposition}\label{prop:oblivious-general}
%	Given any non-oblivious $k$-party protocol $\pi$ with communication $s$ bits for function $f$, an oblivious $k$-party protocol $\pi'$ for $f$ with at most $ks$ total communication can be constructed.
%\end{proposition}
%
%\begin{proof}
%	To construct protocol $\pi'$, run protocol $\pi$ with the condition that each player $\player{i}$ for $i \in [k]$ sends a message of length $s$ by padding. This length is sufficient because the non-oblivious protocol $\pi$ has total communication at most $s$ bits. The total communication of protocol $\pi'$ is $ks$ bits. 
%\end{proof}

\begin{definition}\label{def:rand-cc}
	 The \textnormal{\textbf{randomized communication complexity}} of $f$, with probability of error $\delta$,  is defined as the minimum communication cost of any protocol which solves $f$ with probability of success at least $1-\delta$.
\end{definition}

\subsection{Information Theoretic Tools}\label{sec:prelim-info-theory}

Our proof relies on tools from information theory, and we state the basic definitions and the inequalities we need in this section. 
Proofs of the statements and more details can be found in Chapter 2 of a textbook on information theory by Cover and Thomas \cite{CoverT06}.

\begin{definition}[Shannon Entropy]\label{def:entropy}
	For any random variable $\rX$ over support $\cA$, the Shannon entropy of $\rX$, denoted by $\HH(\rX)$ is defined as,
	\[
	\HH(\rX) = \sum_{A \in \cA} \Pr[\rX = A] \cdot \log(1/ \Pr[\rX = A]).
	\]
	For any event $\cE$ we define $\HH(\rX \mid \cE)$ in the same way, as the entropy of distribution of $\rX$ conditioned on the event $\cE$.
	For any two random variables $\rX$ and $\rY$, the entropy of $\rX$ conditioned on $\rY$, denoted by $\HH(\rX \mid \rY)$ is defined as,
	\[
	\HH(\rX \mid \rY) = \Exp_{Y \sim \rY} \HH(\rX \mid \rY = Y).
	\]
\end{definition}

\begin{fact}\label{fact:en-mi-facts}
	We know the following about entropy and mutual information:
	\begin{enumerate}
		\item \label{part:entropy-bounds} 	For any random variable $\rX$, the entropy obeys the bound: $0 \leq \HH(\rX) \leq \log_2(\card{\cX})$ where $\cX$ is the support of $\rX$.
		\item \label{part:cond-reduce-entropy} For any two random variables $\rX, \rY$, $\HH(\rX \mid \rY) \leq \HH(\rX)$ with equality holding iff $\rX \perp \rY$.
		\item \label{part:chain-rule-entropy} Chain Rule of Entropy: For $m \geq 1$ and any tuple of random variables $\rX = (\rX_1, \rX_2, \ldots, \rX_m)$, $\en{\rX}= \sum_{i \in [m]} \en{\rX_i \mid \rX_{<i}}$.
		\item \label{part:entropy-subadditive} Subadditivity of entropy: For $m \geq 1$ and any tuple of random variables $\rX = (\rX_1, \rX_2, \ldots, \rX_m)$, $\en{\rX}\leq  \sum_{i \in [m]} \en{\rX_i}$.
		%\item \label{part:mi-independent} For any $\rX, \rY$, $\mi{\rX}{\rY} \geq 0$ with equality holding iff $\rX \perp \rY$. 
		%\item \label{part:cond-define} For any three random variables $\rX, \rY, \rZ$, $\mi{\rX}{\rY \mid \rZ} = \Exp_{Z \sim \rZ} \mi{\rX}{\rY \mid \rZ = Z}$. 
		%\item \label{part:mi-cond-event-remove} For any three random variables $\rX, \rY,\rZ$ and an event $\cE$, $\mi{\rX}{\rY \mid \rZ, \cE} = \mi{\rX}{\rY \mid \rZ}$ iff the joint distribution of $(\rX, \rY, \rZ) $ is independent of $\cE$.  
		%\item \label{part:chain-rule-mi} Chain Rule of Mutual Information: For $m \geq 1$, random variables $\rY, \rZ$, and any tuple of random variables $\rX = (\rX_1, \rX_2, \ldots, \rX_m)$, $\mi{\rX}{\rY \mid \rZ} = \sum_{i \in [m]} \mi{\rX_i}{\rY \mid \rX_1, \rX_2, \ldots, \rX_{i-1},\rZ}$.
	\end{enumerate}
\end{fact}

%We can define mutual information using entropy now.
%
%\begin{definition}[Mutual Information]\label{def:mutual-information}
%	The mutual information between two random variables $\rX$ and $\rY$, denoted by $\mi{\rX}{\rY}$ is defined as,
%	\[
%	\mi{\rX}{\rY} = \HH(\rX) - \HH(\rX \mid \rY).
%	\]
%	The conditional mutual information between $\rX$ and $\rY$ conditioned on random variable $\rZ$, denoted by $\mi{\rX}{\rY \mid \rZ}$, is defined as 
%	\[
%	\mi{\rX}{\rY \mid \rZ} = \HH(\rX \mid \rZ) - \HH(\rX \mid \rY, \rZ).		
%	\]
%\end{definition}

We also need the following proposition, which relates entropy to the probability of correctness while estimating a random variable. 

\begin{proposition}[Fano's inequality]\label{prop:fanos}
	Given a binary random variable $\rX$ and an estimator random variable $Y$ and a function $g$ such that $g(Y) = X$ with probability at least $1-\delta$ for $\delta < 1/2$, 
	\[
		\en{\rX \mid \rY} \leq H_2(\delta).
	\]
\end{proposition}

This concludes our preliminaries section.

\newcommand{\rvM}{\rv{M}}
\newcommand{\rvind}{\ensuremath{\bm{\ind}}}
\renewcommand{\Prot}{\ensuremath{\Gamma}}
\renewcommand{\prot}{\ensuremath{\gamma}}
\newcommand{\protindex}{\ensuremath{\pi_{\textsc{index}}}}
\newcommand{\cDindex}{\ensuremath{\cD_{\textsc{index}}}}
\newcommand{\Mindex}{\ensuremath{M_{\textsc{index}}}}
\newcommand{\rMindex}{\ensuremath{\rv{M}_{\textsc{index}}}}
\newcommand{\rT}{\rv{T}}
\newcommand{\rS}{\rv{S}}

\section{The Lower Bound}\label{sec:lb}

In this section, we will prove our lower bound of $\Omega(n)$ on the communication complexity of the $\chain_{n,k}$ problem for $k  = o(n/\log n)$. 
Let us formally define the $\chain_{n,k}$ communication problem first.

\begin{definition}\label{def:chain}
	The $\chain_{n,k}$ communication problem is defined as follows. 
	Given $k+1$ players $\player{i}$ for $i \in [k+1]$  where,
	\begin{itemize}
		\item $\player{i}$ has a string $X_i \in \set{0,1}^n$ for each $i \in [k]$, and,
		\item $\player{i}$ for $1 < i \leq k+1$ has an index $\ind_{i-1} \in [n]$,
	\end{itemize}
	such that,
	\[
	X_i(\ind_{i}) = z,
	\]
	for some bit $z \in \set{0,1}$. 
	The players have a blackboard visible to all the parties. For $i \in [k]$ in ascending order, $\player{i}$ sends a single message $M_i$, after which the index $\ind_i$ is revealed to the blackboard at no cost. $\player{k+1}$ has to output whether $z$ is 0 or 1. 
	Refer to \Cref{fig:board-chain} for an illustration. 
\end{definition}

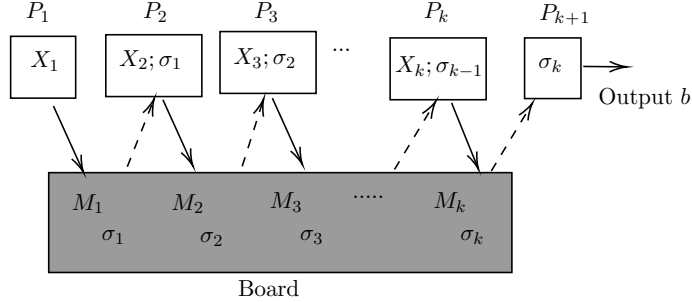
\begin{figure}[h!]
	\centering
	\scalebox{.8}{\tikzset{every picture/.style={line width=0.75pt}} %set default line width to 0.75pt        

\begin{tikzpicture}[x=0.75pt,y=0.75pt,yscale=-1,xscale=1]
	%uncomment if require: \path (0,409); %set diagram left start at 0, and has height of 409
	
	%Shape: Rectangle [id:dp24848550646024137] 
	\draw   (132,131) -- (173,131) -- (173,170) -- (132,170) -- cycle ;
	%Shape: Rectangle [id:dp8447033152047732] 
	\draw   (192,130) -- (252.55,130) -- (252.55,169) -- (192,169) -- cycle ;
	%Shape: Rectangle [id:dp6500157657767858] 
	\draw  [fill={rgb, 255:red, 155; green, 155; blue, 155 }  ,fill opacity=0.23 ] (156,217) -- (448,217) -- (448,280) -- (156,280) -- cycle ;
	%Straight Lines [id:da8633915739116707] 
	\draw    (228,171) -- (246.18,211.18) ;
	\draw [shift={(247,213)}, rotate = 245.66] [color={rgb, 255:red, 0; green, 0; blue, 0 }  ][line width=0.75]    (10.93,-3.29) .. controls (6.95,-1.4) and (3.31,-0.3) .. (0,0) .. controls (3.31,0.3) and (6.95,1.4) .. (10.93,3.29)   ;
	%Straight Lines [id:da7301760621134579] 
	\draw    (410,173) -- (426.25,213.15) ;
	\draw [shift={(427,215)}, rotate = 247.96] [color={rgb, 255:red, 0; green, 0; blue, 0 }  ][line width=0.75]    (10.93,-3.29) .. controls (6.95,-1.4) and (3.31,-0.3) .. (0,0) .. controls (3.31,0.3) and (6.95,1.4) .. (10.93,3.29)   ;
	%Shape: Rectangle [id:dp5694626523829054] 
	\draw   (456,132) -- (491.55,132) -- (491.55,170.27) -- (456,170.27) -- cycle ;
	%Straight Lines [id:da6230485088391524] 
	\draw  [dash pattern={on 4.5pt off 4.5pt}]  (205.55,213.27) -- (221.81,172.13) ;
	\draw [shift={(222.55,170.27)}, rotate = 111.57] [color={rgb, 255:red, 0; green, 0; blue, 0 }  ][line width=0.75]    (10.93,-3.29) .. controls (6.95,-1.4) and (3.31,-0.3) .. (0,0) .. controls (3.31,0.3) and (6.95,1.4) .. (10.93,3.29)   ;
	%Straight Lines [id:da3426735819978517] 
	\draw    (297,170) -- (315.18,210.18) ;
	\draw [shift={(316,212)}, rotate = 245.66] [color={rgb, 255:red, 0; green, 0; blue, 0 }  ][line width=0.75]    (10.93,-3.29) .. controls (6.95,-1.4) and (3.31,-0.3) .. (0,0) .. controls (3.31,0.3) and (6.95,1.4) .. (10.93,3.29)   ;
	%Straight Lines [id:da8144325130991623] 
	\draw  [dash pattern={on 4.5pt off 4.5pt}]  (278,213) -- (290.95,172.18) ;
	\draw [shift={(291.55,170.27)}, rotate = 107.59] [color={rgb, 255:red, 0; green, 0; blue, 0 }  ][line width=0.75]    (10.93,-3.29) .. controls (6.95,-1.4) and (3.31,-0.3) .. (0,0) .. controls (3.31,0.3) and (6.95,1.4) .. (10.93,3.29)   ;
	%Straight Lines [id:da41276681944163696] 
	\draw  [dash pattern={on 4.5pt off 4.5pt}]  (374,215) -- (398.97,173.71) ;
	\draw [shift={(400,172)}, rotate = 121.16] [color={rgb, 255:red, 0; green, 0; blue, 0 }  ][line width=0.75]    (10.93,-3.29) .. controls (6.95,-1.4) and (3.31,-0.3) .. (0,0) .. controls (3.31,0.3) and (6.95,1.4) .. (10.93,3.29)   ;
	%Straight Lines [id:da8129096668812823] 
	\draw  [dash pattern={on 4.5pt off 4.5pt}]  (435,216) -- (459.97,174.71) ;
	\draw [shift={(461,173)}, rotate = 121.16] [color={rgb, 255:red, 0; green, 0; blue, 0 }  ][line width=0.75]    (10.93,-3.29) .. controls (6.95,-1.4) and (3.31,-0.3) .. (0,0) .. controls (3.31,0.3) and (6.95,1.4) .. (10.93,3.29)   ;
	%Straight Lines [id:da14789038258631293] 
	\draw    (159,175) -- (177.18,215.18) ;
	\draw [shift={(178,217)}, rotate = 245.66] [color={rgb, 255:red, 0; green, 0; blue, 0 }  ][line width=0.75]    (10.93,-3.29) .. controls (6.95,-1.4) and (3.31,-0.3) .. (0,0) .. controls (3.31,0.3) and (6.95,1.4) .. (10.93,3.29)   ;
	%Shape: Rectangle [id:dp26583815606850414] 
	\draw   (264,128) -- (324.55,128) -- (324.55,167) -- (264,167) -- cycle ;
	%Shape: Rectangle [id:dp9010535462849016] 
	\draw   (371.23,132) -- (431.77,132) -- (431.77,171) -- (371.23,171) -- cycle ;
	%Straight Lines [id:da6496377348212914] 
	\draw    (493,150) -- (518.55,150.25) ;
	\draw [shift={(520.55,150.27)}, rotate = 180.57] [color={rgb, 255:red, 0; green, 0; blue, 0 }  ][line width=0.75]    (10.93,-3.29) .. controls (6.95,-1.4) and (3.31,-0.3) .. (0,0) .. controls (3.31,0.3) and (6.95,1.4) .. (10.93,3.29)   ;
	
	% Text Node
	\draw (143,139) node [anchor=north west][inner sep=0.75pt]   [align=left] {$\displaystyle X_{1}$};
	% Text Node
	\draw (199,138) node [anchor=north west][inner sep=0.75pt]   [align=left] {$\displaystyle X_{2} ;\sigma _{1}$};
	% Text Node
	\draw (333,137) node [anchor=north west][inner sep=0.75pt]   [align=left] {$\displaystyle ...$};
	% Text Node
	\draw (140,108) node [anchor=north west][inner sep=0.75pt]   [align=left] {$\displaystyle P_{1}$};
	% Text Node
	\draw (214,108) node [anchor=north west][inner sep=0.75pt]   [align=left] {$\displaystyle P_{2}$};
	% Text Node
	\draw (391,108) node [anchor=north west][inner sep=0.75pt]   [align=left] {$\displaystyle P_{k}$};
	% Text Node
	\draw (274,283) node [anchor=north west][inner sep=0.75pt]   [align=left] {Board};
	% Text Node
	\draw (464,137) node [anchor=north west][inner sep=0.75pt]   [align=left] {$ $};
	% Text Node
	\draw (464,109) node [anchor=north west][inner sep=0.75pt]   [align=left] {$\displaystyle P_{k}{}_{+}{}_{1}$};
	% Text Node
	\draw (284,108) node [anchor=north west][inner sep=0.75pt]   [align=left] {$\displaystyle P_{3}$};
	% Text Node
	\draw (271,137) node [anchor=north west][inner sep=0.75pt]   [align=left] {$\displaystyle X_{3} ;\sigma _{2}$};
	% Text Node
	\draw (373.23,142) node [anchor=north west][inner sep=0.75pt]   [align=left] {$\displaystyle X_{k} ;\sigma _{k-1}$};
	% Text Node
	\draw (463.23,143) node [anchor=north west][inner sep=0.75pt]   [align=left] {$\displaystyle \sigma _{k}$};
	% Text Node
	\draw (169,229) node [anchor=north west][inner sep=0.75pt]   [align=left] {$\displaystyle M_{1}$};
	% Text Node
	\draw (188,252) node [anchor=north west][inner sep=0.75pt]   [align=left] {$\displaystyle \sigma _{1}$};
	% Text Node
	\draw (232,229) node [anchor=north west][inner sep=0.75pt]   [align=left] {$\displaystyle M_{2}$};
	% Text Node
	\draw (250,253) node [anchor=north west][inner sep=0.75pt]   [align=left] {$\displaystyle \sigma _{2}$};
	% Text Node
	\draw (294,228) node [anchor=north west][inner sep=0.75pt]   [align=left] {$\displaystyle M_{3}$};
	% Text Node
	\draw (313,252) node [anchor=north west][inner sep=0.75pt]   [align=left] {$\displaystyle \sigma _{3}$};
	% Text Node
	\draw (397,228) node [anchor=north west][inner sep=0.75pt]   [align=left] {$\displaystyle M_{k}$};
	% Text Node
	\draw (346,230) node [anchor=north west][inner sep=0.75pt]   [align=left] {$\displaystyle .....$};
	% Text Node
	\draw (414,252) node [anchor=north west][inner sep=0.75pt]   [align=left] {$\displaystyle \sigma _{k}$};
	% Text Node
	\draw (501,162) node [anchor=north west][inner sep=0.75pt]   [align=left] {Output $\displaystyle b$};

\end{tikzpicture}}\caption{An illustration of the $\chain_{n,k}$ problem from \Cref{def:chain}. The solid arrows illustrate that player $\player{i}$ writes a message $M_i$ to the board. The dashed arrows indicate that $\player{i}$ can read the contents of the board. It also shows the order in which the messages are sent by the players and indices are released. } \label{fig:board-chain}
\end{figure}

Let us recall the statement of our main result. 

\begin{restate}[\Cref{thm:main}]
For any $n, k \geq 1$, any protocol for $\chain_{n,k}$ with probability of success at least 2/3, requires $\Omega(n - k \log n)$ total bits of communication.
\end{restate}

We give our hard distribution for $\chain_{n,k}$ in \Cref{subsec:lb-define} and give the proof of \Cref{thm:main} in \Cref{subsec:proof-lb} except for the analysis of biased index, which is given in \Cref{subsec:bias-ind}. Lastly, we extend the arguments to Augmented Chain in \Cref{subsec:aug-chain}. 

%Note that when $k = o(n/\log n)$, this gives the optimal lower bound of $\Omega(n)$. %This directly gives a lower bound of $\Omega(n/k-\log n)$ for non-oblivious protocols when $k = o(n/\log n)$ by \Cref{prop:oblivious-general}.

%By a black-box reduction from non-oblivious protocols to oblivious protocols with a multiplicative $k$ factor in communication, \Cref{thm:main} also gives an $\Omega(n/k)$ lower bound for the total communication in non-oblivious protocols. 

\subsection{Setting Up the Problem}\label{subsec:lb-define}

In this subsection, we start by defining the notation for our proof, and we describe the input distributions to $\chain_{n,k}$.

The input hard distribution is as follows. Let $\cLhalf \subset \{0,1\}^n$ be the subset of strings where the number of ones is exactly equal to $n/2$.

\begin{tbox}
	Distribution $\cD$ for $\chain_{n,k}$:
	\begin{enumerate}
		\item Pick a bit $z$ uniformly at random from $\set{0,1}$.
		\item For each $i \in [k]$, sample $(X_i, \ind_i)$ uniformly at random from $\cLhalf \times [n]$ and independently conditioned on $X_i(\ind_i) = z$. 
	\end{enumerate}
\end{tbox}

\paragraph{Notation.}
We use $\rX_i$ to denote the random variable corresponding to string $X_i$ and $\rvind_i$ to denote the random variable corresponding to index $\ind_i$ for $i \in [k]$.  To denote the random variable corresponding to the first $i-1$ strings and indices, we use $\rX_{<i}$ and $\rvind_{<i}$ respectively. 

We use $\rM_i$ to denote the random variable corresponding to the message $M_i$ sent by $\player{i}$ to the blackboard. We use $M =  (M_1, M_2, \ldots, M_k)$ to denote the tuple containing the messages of all the players, and $\rM$ to denote the random variable corresponding to $M$. 

Let $\pi$ be a deterministic protocol for $\chain_{n,k}$ with probability of success at least $2/3$ when the input is distributed according to $\cD$.  We use $\Prot$ to denote the random variable corresponding to the contents of the blackboard (referred to as a \textbf{transcript}), and we use $\prot$ to also denote transcripts sampled from $\Prot$. We use $\Prot_i$ to denote the random variable of the tuple $(M^i, \ind^i)$ for $i \in [k]$. 
We use $\pi(\prot_{<i}, X_{i})$ to denote the output of $\player{i}$ when the contents of the blackboard are $\prot_{<i}$ and the input is $X_i$ for $i \in [k]$. 
%(Note that $\Prot_i$ does not correspond to the message of $\player{i}$ alone, $\Prot_i$ corresponds to both $M^i$ and $\ind^i$ where $\ind^i$ is held by $\player{i+1}$.)

Let $s$ be the total length of all the messages sent by the players in $\prot$. We assume that the total length of the messages is exactly $s$ by padding. For any random variable $\rA$, we use $\cD(\rA)$ to denote the distribution of the random variable $\rA$, as the input is distributed according to $\cD$. We replace $\cD(\rA \mid \rB = b)$ with $\cD(\rA \mid b)$ for ease of readability whenever it is clear from context. 

Let $\rZ$ denote the random variable corresponding to bit $z$. 
The bit $z$ corresponds to the answer to $\chain_{n,k}$. 
We show the lower bound of $\Omega(n-k \log n)$ for distinguishing between the case when $z = 0$ and $z = 1$ based on the contents of the blackboard.

We need one important observation about the distribution $\cD$. 

\begin{observation}\label{obs:indep-sigma-Z}
	For any $i \in [k]$, random variable $\rX_{i}, \rvind_{i}$ is independent of $\Prot_{< i}$ conditioned on $\rZ$. 
\end{observation}
\begin{proof}
	For any fixed value of $\rZ$, $\rX_{i}, \rvind_{i}$ are chosen uniformly at random from $\cLhalf \times [n]$ such that $\rX_{i}(\rvind_{i}) = \rZ$. This choice is independent of any $\rX_j, \rvind_j$ with $j \neq i$, and thus independent of $\Prot_{<i}$. 
\end{proof}

We are ready to proceed with the proof of our main theorem.
%
%\begin{observation}\label{obs:ind-x}
%	Random variable $\rZ$ is independent of $\rX^{i+1}, \rM^{i+1}$ conditioned on $\Prot_{\leq i}$.
%\end{observation}
%\begin{proof}
%	The distribution of $\rX^{i+1}$ is always uniform over strings in $\cLhalf$, and $\rM^{i+1}$ is fixed by $\Prot_{\leq i}$ and $\rX^{i+1}$. This distribution is independent of $\rZ$. Note, however, that $\rZ$ is not independent of $\rX^{i+1}, \rvind^{i+1}$.
%\end{proof}

\subsection{Proof of Lower Bound}\label{subsec:proof-lb}

We start by showing that in any successful protocol, the entropy of the distribution of $\rZ$ conditioned on the message must be small.
\begin{claim}[The transcript reveals information about $\rZ$.]\label{clm:entropy-low}
	\[
	\en{\rZ \mid \Prot} \leq 24/25.
	\]
\end{claim}

\begin{proof}
	We know that $\player{k+1}$ successfully finds the value of $z$ with probability of success at least 2/3 using the transcript. Thus, using Fano's inequality in \Cref{prop:fanos}, we have,
	\[
	\en{\rZ \mid \Prot} \leq H_2(2/3) \leq 24/25. \qedhere
	\]
\end{proof}

The main part of the proof is that we show a lower bound the entropy of $\rZ$ conditioned on the transcripts using the entropy of the message.

\begin{lemma}\label{lem:dist-ub} For any protocol $\pi$,
	\[
	1-\frac{12}n \cdot (\en{\rM}+ k\log n) \leq  \en{\rZ \mid \Prot}.
	\]
\end{lemma}

Before we prove \Cref{lem:dist-ub}, we can easily show that it implies \Cref{thm:main}.

\begin{proof}[Proof of \Cref{thm:main}]
	Combining \Cref{lem:dist-ub} with \Cref{clm:entropy-low}, we get,
	\[
	1-\frac{12}n \cdot ( \en{\rM}+ k\log n) \leq \en{\rZ \mid \Prot} \leq \frac{24}{25}.
	\] 
	This gives that, 
	\[
	\en{\rM} \geq \frac{n}{25 \cdot 12} - k \log n.
	\]
	We know from \Cref{fact:en-mi-facts}-(\ref{part:entropy-bounds}) that $\en{\rM} \leq \log (2^s) = s$, which proves that the total number of bits $s = \Omega(n - k \log n)$ for any deterministic protocol.
	By Yao's minimax principle, we get a lower bound of $\Omega(n-k \log n)$ on the randomized communication complexity of $\chain_{n,k}$. 
\end{proof}

The proof of \Cref{lem:dist-ub} employs a reduction to the two player $\indexprob_{n}$. 
However, in these instances of $\indexprob_n$, Alice and Bob already have some partial information about the answer. We call this problem the \textbf{biased index} problem, and it is defined based on parameter $\theta \in [-1/2, 1/2]$, which is the initial bias known about the answer. 

\begin{definition}\label{def:bias-index}
	The \textnormal{\textbf{biased index distributional communication problem}}, denoted by $\biasedind(\theta)$ for $\theta \in [-1/2, 1/2]$, is defined as follows.
	
	Sample $W \in \{0,1\}$ such that $W= 1$ with probability $1/2 + \theta$, and $W = 0$ otherwise.  Sample $(Y, \rho)$ uniformly at random from $\cLhalf \times [n]$ conditioned on $Y(\rho) = W$. Give string $Y$ to Alice, and index $\rho$ to Bob.  Bob has to output $Y(\rho)$ after a single message $\Mindex$ from Alice. 
\end{definition}

Let $\protindex$ be a deterministic protocol for $\biasedind(\theta)$.
Let $\rW, \rY, \rrho, \rMindex$ denote the random variables corresponding to $W, Y$, $\rho$ and $\Mindex$ respectively. Let $\cD_{\theta}$ denote the joint distribution of $\rW, \rY, \rrho$ and $ \rMindex$ in $\biasedind(\theta)$.

We prove the following lemma about $\biasedind(\theta)$ in \Cref{subsec:bias-ind}.

\begin{lemma}[Biased Index]\label{lem:bias-ind}
	For any protocol $\protindex$ for $\biasedind(\theta)$,
	\[
	\en{\rW \mid \rMindex, \rrho} \geq H_2(1/2+\theta) - \frac{2}{n} \cdot(\en{\rMindex} + \log n).
	\]
\end{lemma}

Now we can prove \Cref{lem:dist-ub}. 

\begin{proof}
	We prove that, for any $0 \leq \ell \leq k$,
	\[
	\en{\rZ \mid \Prot_{\leq \ell}}  \geq 1- \frac{2}{n}  \cdot (\sum_{j=1}^{\ell} \en{\rM_j \mid \Prot_{< j}}+ \ell \log n),
	\]
	by induction on $\ell$. This implies the statement of the lemma because when $\ell = k$, 
	\begin{align*}
		\en{\rZ \mid \Prot} &= \en{\rZ \mid \Prot_{\leq k}}  \\ &\geq  1- \frac{2}{n}  \cdot (\sum_{j=1}^{k} \en{\rM_j \mid \Prot_{< j}}+k \log n)  \\
		&=   1- \frac{2}{n}  \cdot (\sum_{j=1}^{k} \en{\rM_j \mid \rM_{<j}, \rvind_{<j}}+k\log n)  \tag{by defintion of random variable $\Prot_{<j}$}\\
		&\geq   1- \frac{2}{n}  \cdot (\sum_{j=1}^{k} \en{\rM_j \mid \rM_{<j}}+k\log n) \tag{as conditioning reduces entropy, \Cref{fact:en-mi-facts}-(\ref{part:cond-reduce-entropy})} \\
		&\geq 1-\frac{2}{n} \cdot (\en{\rM} + k \log n). \tag{by chain rule of entropy, \Cref{fact:en-mi-facts}-(\ref{part:chain-rule-entropy})}
	\end{align*}
	
	The base case when $\ell  = 0$ is trivial, as $\en{\rZ} = 1$ when $\rZ$ is uniform over $\{0,1\}$, by \Cref{fact:en-mi-facts}-(\ref{part:entropy-bounds}). Let us assume that the statement is true for some $0 \leq i -1< k$ for the induction hypothesis, and prove it for $i$. 
	
	For any $i \in [k]$ and any $\prot_{< i}$ sampled from $\cD(\Prot_{< i})$, we define
	$\theta^{(\prot_{< i})}$ to be the bias of distribution of $\rZ$ conditioned on $\prot_{< i}$ from uniform, that is,
	\[
	\theta^{(\prot_{< i})} = \Pr(Z = 1 \mid \Prot_{< i} = \prot_{< i})-1/2.
	\]
	We also define distribution $\cD^{(\prot_{< i})}$ as the joint distribution of $\rZ, \rX_{i}, \rvind_{i}$ and $\Prot_{i}$ conditioned on the event $\Prot_{< i} = \prot_{< i}$.

	Now, we employ a reduction to $\biasedind(\theta^{(\prot_{< i})})$ problem. Alice takes on the role of $\player{i}$ and Bob takes on the role of $\player{i+1}$ when the contents of the blackboard are fixed to be $\prot_{<i}$. 
	We will argue that distribution $\cD^{(\prot_{<i})}(\rZ, \rX_i, \rvind_i, \rM_i)$ of random variables in $\chain_{n,k}$ is identical to the distribution $\cD_{\theta^{(\prot_{<i})}}(\rW, \rY, \rrho, \rMindex)$ of the random variables in $\biasedind(\theta^{(\prot_{<i})})$. 
	
	\begin{enumerate}[label=$(\roman*)$]
		\item Random variable $\rZ$ is set to be 1 with probability $1/2 + \theta^{(\prot_{<i})}$ and 0 otherwise, same as random variable $\rW$ in $\cD_{\theta^{(\prot_{<i})}}$. 
		\item Random variable $\rY, \rrho$ is sampled from $\cLhalf \times [n]$ uniformly at random conditioned on $\rY(\rrho) = \rZ$, and this is the same as the distribution of $\rX_{i}, \rvind_i$ in $\cD^{(\prot_{<i})}$.  This is because $\rX_i, \rvind_i$ are independent of $\Prot_{<i} = \prot_{<i}$ conditioned on $\rZ$, by \Cref{obs:indep-sigma-Z}.
		\item As Alice takes on the role of $\player{i}$, given the same input, she sends the same message as $\player{i}$ when the contents of the blackboard are $\prot_{<i}$. The distribution of $\rM_i$ is the same as that of $\rMindex$, as distribution of $\rX_{i}, \rvind_{i}$ is the same as that of $\rY, \rrho$.  
	\end{enumerate}
	
	Hence using \Cref{lem:bias-ind}, we get,
	\begin{equation}\label{eq:from-ind-1}
		\en{\rZ \mid \Prot_{i}, \Prot_{< i} = \prot_{< i}} \geq H_2(1/2 + \theta^{(\prot_{< i})}) - \frac{2}{n} \cdot (\en{\rM_{i} \mid \Prot_{< i} = \prot_{< i}} + \log n).
	\end{equation}
	We also observe that, 
	\begin{equation}\label{eq:bin-ent-entZ}
		H_2(1/2 + \theta^{(\prot_{< i})}) = \en{\rZ \mid \Prot_{< i} = \prot_{< i} },
	\end{equation}
	by definition of binary entropy and \Cref{def:entropy}. 
	
	Lastly, we have,
	\begin{align*}
		\en{\rZ \mid \Prot_{\leq i}} &= \Exp_{\prot_{< i} \sim \cD(\Prot_{< i})} \en{\rZ \mid \Prot_{i}, \Prot_{< i} = \prot_{< i}} \tag{by definition of entropy, \Cref{def:entropy}} \\
		&\geq  \Exp_{\prot_{< i} \sim \cD(\Prot_{< i})}\big[ H_2(1/2 + \theta^{(\prot_{< i})})   -\frac{2}{n} \cdot (\en{\rM_{i} \mid \Prot_i = \prot_{< i}} + \log n)\big] \tag{by \Cref{eq:from-ind-1}}\\
		&= \Exp_{\prot_{< i} \sim \cD(\Prot_{< i})} \big[ \en{\rZ \mid \Prot_{< i} = \prot_{< i}}-\frac{2}{n} \cdot (\en{\rM_{i} \mid \Prot_i = \prot_{< i}} + \log n)\big] \tag{by \Cref{eq:bin-ent-entZ}}\\
		&= \en{\rZ \mid \Prot_{<i}} -\frac2n \cdot \paren{ \Exp_{\prot_{< i} \sim \cD(\Prot_{< i})} (\en{\rM_{i} \mid \Prot_{< i} = \prot_{< i}} + \log n)} \\
		&= \en{\rZ \mid \Prot_{<i }} -\frac2n \cdot(\en{\rM_{i} \mid \Prot_{< i}} + \log n)  \\
		&\geq 1-\frac{2}{n} \cdot \paren{\sum_{j=1}^{i} \en{\rM_j \mid \Prot_{< j}} + i \log n},\tag{by the induction hypothesis}
	\end{align*}
	finishing the proof.
\end{proof}

\subsection{Biased Index}\label{subsec:bias-ind}

In this subsection, we will prove \Cref{lem:bias-ind}. Let us first recall the input distribution to $\biasedind(\theta)$. 

\begin{tbox}
	\textbf{Distribution} $\cD_{\theta}$:
	Sample $W = 1$ with probability $1/2 + \theta$ and set $W  = 0$ otherwise.  
	Sample $(Y, \rho) \in \cLhalf \times [n]$ conditioned on $Y(\rho) = W$. 
\end{tbox}

In $\cD_{\theta}$, the distribution of $\rY$ and $\rrho$ are highly correlated. We give an alternate way of sampling $Y, \rho$ so that this correlation is removed partially. 

\begin{tbox}
	\textbf{Distribution} $\cD_{\theta}'$:
	
	For $\theta \geq 0$:
	\begin{enumerate}[label=$(\roman*)$]
		\item Sample set $T \subset [n]$ of size $b = n/(1+2\theta)$ uniformly at random. 
		\item Sample a set $S$ of  $n/2$ indices  from $T$ uniformly at random and set them to 1, and set $[n]\setminus S$  to 0 to get $Y$. 
		\item Sample $\rho$ by sampling an index uniformly at random from $T$.
	\end{enumerate}
	For $\theta<0$:
	\begin{enumerate}[label=$(\roman*)$]
		\item Sample set $T \subset [n]$ of size $b = n/(1-2\theta)$ uniformly at random. 
		\item Sample a set $S$ of  $n/2$ indices  from $T$ uniformly at random and set them to 0, and set $[n]\setminus S$  to 1 to get $Y$. 
		\item Sample $\rho$ by sampling an index uniformly at random from $T$.
	\end{enumerate}
\end{tbox}
In this section, we assume that $\theta\geq 0$. For the case when $\theta < 0$, the proof follows in the same vein, and is not presented. 
Let $\rT, \rS$ denote the random variables corresponding to set $T$ and set $S$ respectively. 
We will show that distributions $\cD_{\theta}$ and $\cD_{\theta}'$ are in fact identical. 

\begin{claim}\label{clm:prob-sample-D1}
	In distribution $\cD_{\theta}$, for any $(Y, \rho) \in \cLhalf \times [n]$, 
	\begin{equation*}
		\Pr(\rY = Y, \rrho = \rho) =\begin{cases}
			\frac{1+2\theta}{n \cdot \binom{n}{n/2}} &\textnormal{when $Y(\rho) = 1$}, \\
			\frac{1-2\theta}{n \cdot \binom{n}{n/2}} &\textnormal{when $Y(\rho)  = 0$.}
		\end{cases}
	\end{equation*}
\end{claim}

\begin{proof}
	For any bit $c \in \{0,1\}$, when $Y(\rho) = c$, 
	\begin{align*}
		\Pr(\rY = Y, \rrho = \rho) &= \Pr(\rW =c ) \cdot \Pr(\textnormal{$(Y, \rho)$ is sampled conditioned on $\rW = c$}) \\
		&= \Pr(\rW = c)\cdot \frac2{\binom{n}{n/2} \cdot n} \tag{as $Y, \rho$ are chosen uniformly at random with $Y(\rho) = c$} \\
		&= \frac{2 \cdot \Pr(\rW = c)}{\binom{n}{n/2} \cdot n}. 
	\end{align*}
	Substituting $\Pr(\rW = 1)$ and $\Pr(\rW = 0)$ as $1/2 + \theta$ and $1/2 - \theta $ respectively proves the claim. 
\end{proof}

\begin{claim}\label{clm:cD-two-same}
	Distribution $\cD_{\theta}$ is the same as $\cD_{\theta}'$. 
\end{claim}

\begin{proof}
	For a fixed $Y, \rho$ we will show that the probability of sampling $(Y, \rho)$ is same in both $\cD_{\theta}$ and $\cD_{\theta}'$. 
	Let $S'$ be the set of indices where the string $Y$ takes value 1. 
	
	We begin with the case when $Y(\rho) = 1$ and $\rho \in S'$.
	In distribution $\cD_{\theta}'$, for $(Y, \rho)$ to be sampled, the random variable $\rS $ must be fixed to $S'$, and index $\rho$ must be chosen from $\rT$.
	
	Let us find the probability with which $\rS$ is chosen to be $S'$.  The distribution of  $\rS$ is uniform over all subsets of $[n]$ of size $n/2$, as $\rT$ is uniform, and $\rS$ is chosen uniformly from $\rT$. Therefore, 
	\begin{align*}
		\Pr(\rS = S') &= \frac1{\binom{n}{n/2}}. 
	\end{align*}
	Now, we can find the probability that $(Y, \rho)$ is sampled. 
	\begin{align*}
		\Pr((Y, \rho) \textnormal{ is sampled from $\cD_{\theta}'$}) 	 &= \Pr( \rS = S') \cdot \Pr(\textnormal{$\rho \in S'$ is chosen from $\rT$} \mid \rS = S')  \\
		&= \frac1{\binom{n}{n/2}} \cdot  \Pr(\textnormal{$\rho$ is chosen from $\rT$} \mid  \rho \in \rT) \tag{as $Y(\rho) = 1$ and $\rho \in S' = \rS \subset \rT$} \\
		&= \frac1{\binom{n}{n/2}} \cdot \frac1{b} \tag{choosing one index among $b$ choices}\\
		&=  \frac{(1+2\theta)}{\binom{n}{n/2} \cdot n}.
	\end{align*}
	
	Now we look at the case when $Y(\rho) = 0$ and $\rho \notin S'$. Here, for $(Y, \rho)$ to be sampled, we first need that $\rS = S'$ and one index from the $b - n/2 $ remaining indices of $\rT$ must be $\rho$. 
	\begin{align*}
		\Pr(S' = \rS \textnormal{ and } \rho \in \rT) &= \Pr(S' \cup \set{\rho} \subset \rT) \cdot \Pr(S' = \rS \mid S' \subset \rT) \\
		&=  \binom{n/2-1}{b-n/2-1} \cdot \frac1{\binom{n}{b}} \cdot \Pr(S' = \rS \mid S' \subset \rT)  \tag{as we choose remaining elements of $\rT$ from $[n] \setminus (S' \cup \set{\rho})$} \\
		&= \binom{n/2-1}{b-n/2-1} \cdot \frac1{\binom{n}{b}}  \cdot \frac1{\binom{b}{n/2}} \tag{choosing one choice among all for $\rS$} \\
		&=  \frac{(b-n/2)}{n/2} \cdot \frac{(n/2)!}{(b-n/2)! (n-b)!} \cdot \frac{b!  (n-b)!}{n!} \cdot \frac{(n/2)! (b-n/2)!}{b!} \\
		&= (2b-n) \cdot \frac1{n \cdot \binom{n}{n/2}}
	\end{align*}
	The final probability of sampling $(Y, \rho)$ is,
	\begin{align*}
		\Pr((Y, \rho) \textnormal{ is sampled from $\cD_{\theta}'$}) 	 &= 	\Pr(S' = \rS \textnormal{ and } \rho \in \rT)  \cdot  \Pr(\textnormal{$\rho$ is chosen from $\rT$} \mid  \rho \in \rT) \\
		&=  \frac{2b-n}{b} \cdot \frac1{n \cdot \binom{n}{n/2}} \\
		&= \frac{1-2\theta}{n \cdot \binom{n}{n/2}}.
	\end{align*}
	Using \Cref{clm:prob-sample-D1} we complete the proof. 
\end{proof}

Using this alternate way of sampling, it is easy to see that random variables $\rY$ and $\rrho$ are independent of each other conditioned on $\rT$. It can also be extended to include random variable $\rMindex$, as it is only a function of $\rY$. 
\begin{observation}\label{obs:ind-cond-M-T}
	In distribution $\cD_{\theta}'$, conditioned on $\rT = T$,  for any $i \in T$, distribution of random variables $\rY, \rMindex$ is independent of event $\rrho = i$. 
\end{observation}
\begin{proof}
	Conditioned on $\rT = T$, string $Y$ is chosen by picking an $n/2$ size set $S$ uniformly at random from $T$, and setting these indices to 1, and this choice also fixes $\rMindex$ as the protocol is deterministic. And index $\rho$ is chosen uniformly at random from $T$, independently of the choice of $S$, by definition of $\cD_{\theta}'$. Thus, choice of $\rrho = i$ is independent of random variables $\rY, \rMindex$. 
\end{proof}

Next, we will show that even conditioned on $\rT$, the entropy of $\rY$ remains large.

\begin{claim}\label{clm:entropy-large}
	\[
	\en{\rY \mid \rT} \geq \frac{n}{(1+2\theta)} \cdot H_2(1/2+\theta) - 2\log n.
	\]
\end{claim}

\begin{proof}
	We assume that $\theta < 1/2$, as otherwise, the statement is vacuously true. $H_2(1) = 0$ by definition, and entropy is always non-negative by \Cref{fact:en-mi-facts}-(\ref{part:entropy-bounds}).
	
	Conditioned on $\rT = T$, we know that $\rY$ is fixed by choosing set $\rS$ uniformly at random. Thus,
	\begin{align*}
		\en{\rY \mid \rT} &= \log \paren{\binom{b}{n/2}} \tag{by \Cref{fact:en-mi-facts}-(\ref{part:entropy-bounds})} \\
		&\geq \log \paren{2^{b H_2(n/2b)} \cdot \sqrt{\frac{b}{8(n/2)(b-n/2)}}} \tag{by \Cref{fact:binom-lower-bound}, and $n/2 < b$, as $\theta < 1/2$} \\
		&= b \cdot H_2(n/2b) + \frac12 \cdot \log (\frac{b}{8(n/2)(b-n/2)}) \\
		&= \frac{n}{(1+2\theta)} \cdot H_2(1/2 + \theta) + \frac12 \log (\frac1{4n \cdot (1/2-\theta)}) \\
		& \geq  \frac{n}{(1+2\theta)} \cdot H_2(1/2 + \theta) + \frac12 \log (1/4n) \tag{as $1/2-\theta \leq 1$} \\
		&\geq \frac{n}{(1+2\theta)} \cdot H_2(1/2 + \theta) - 2 \log n. \qedhere
	\end{align*}
\end{proof}

We are ready to prove \Cref{lem:bias-ind}.
\begin{proof}[Proof of \Cref{lem:bias-ind}]
	We can lower bound the entropy of $\rW$ conditioned on $\rMindex, \rrho$ as, 
	\begin{align*}
		\en{\rW \mid \rMindex, \rrho} &\geq \en{\rW \mid \rMindex, \rrho, \rT} \tag{as conditioning reduces entropy, \Cref{fact:en-mi-facts}-(\ref{part:cond-reduce-entropy})} \\
		&= \Exp_{\rT = T} \big[\frac1{b} \cdot \sum_{\rho \in T} \en{\rW \mid \rMindex, \rrho = \rho, \rT = T}\big] \tag{as $\rrho$ is uniform over $T$} \\
		&= \Exp_{\rT = T}  \big[\frac1{b} \cdot \sum_{\rho \in T} \en{\rY(\rho) \mid \rMindex, \rrho = \rho, \rT = T}\big] \tag{as $W = Y(\rho)$ by definition of $\cD_{\theta}$ }\\
		&=  \Exp_{\rT = T}  \big[\frac1{b} \cdot \sum_{\rho \in T} \en{\rY(\rho) \mid \rMindex, \rT = T}\big] \tag{as $\rY, \rMindex \perp (\rrho = \rho) \mid \rT = T$, by \Cref{obs:ind-cond-M-T}} \\
		&\geq \Exp_{\rT = T}  \big[\frac1{b} \cdot  \en{\rY(T) \mid \rMindex, \rT = T}\big] \tag{by subadditivity of entropy, \Cref{fact:en-mi-facts}-(\ref{part:entropy-subadditive})} \\
		&=  \Exp_{\rT = T}  \big[\frac1{b} \cdot \en{\rY \mid \rMindex, \rT = T}\big] \tag{as $Y([n] \setminus T)$ is fixed to be 0} \\
		&= \frac1b \cdot \en{\rY \mid \rMindex, \rT}. 
	\end{align*}
	Thus it follows that,
	\begin{equation} \label{eq:inter-new-1}
		\en{\rY \mid \rMindex, \rT}  \leq b \cdot \en{\rW \mid \rMindex, \rrho} .
	\end{equation}
	We also have,
	\begin{align*}
		\en{\rY \mid \rT} &= \en{\rY, \rMindex \mid \rT} \tag{as $M$ is fixed by $Y$} \\
		&= \en{\rMindex \mid \rT} + \en{\rY \mid \rMindex, \rT} \tag{by chain rule of entropy, \Cref{fact:en-mi-facts}-(\ref{part:chain-rule-entropy})} \\
		&\leq \en{\rMindex \mid \rT} + b \cdot \en{\rW \mid \rMindex, \rrho} \tag{by \Cref{eq:inter-new-1}} \\
		&\leq \en{\rMindex} +  b \cdot \en{\rW \mid \rMindex, \rrho}. \tag{as conditioning reduces entropy, by \Cref{fact:en-mi-facts}-(\ref{part:cond-reduce-entropy})} 
	\end{align*}
	Combining with \Cref{clm:entropy-large}, we get,
	\begin{align*}
		\en{\rMindex} +  b \cdot \en{\rW \mid \rMindex, \rrho}   &\geq   b \cdot H_2(1/2 + \theta) - 2\log n \tag{as $b = n/(1+2\theta)$}\\
		b \cdot	\en{\rW \mid \rMindex, \rrho} &\geq b \cdot H_2(1/2 + \theta) - (\en{\rMindex}  + 2 \log n).  \tag{rearranging the terms}
		\end{align*}
	Dividing both sides by $b$, we get,
	\begin{align*}
		\en{\rW \mid \rMindex, \rrho}	&\geq H_2(1/2 + \theta) - \frac1b \cdot (\en{\rMindex} + 2 \log n) \\
		 &\geq  H_2(1/2 + \theta) - \frac2n \cdot (\en{\rMindex} + 2 \log n), \tag{as $b \geq n/2$}
	\end{align*}
	finishing the proof.
\end{proof}

\subsection{Extension to Augmented Chain}\label{subsec:aug-chain}

In this subsection, we will extend our lower bound to the Augmented Chain problem introduced by \cite{DarkDK23}. We begin by defining Augmented Index. 

Augmented Index is a close variant of the $\indexprob$ problem. Here, in addition to having the index $\ind$, Bob also has the bits $X(< \ind)$. We know that this generalization also requires $\Omega(n)$ communication when Alice sends a message to Bob \cite{MiltersenNSW98}. This problem is particularly useful for proving lower bounds for turnstile streams (see e.g., \cite{ClarksonW09,KaneNW10,DarkK20}, and references therein). Tight information cost trade-off for this variant in the two-way communication model was proved by \cite{ChakrabartiCKM13}.

The formal definition of Augmented Chain follows.

\begin{definition}[Augmented Chain]\label{def:aug-chain}
	The $\augchain_{n,k}$ communication problem is defined as follows. 
	Given $k+1$ players $\player{i}$ for $i \in [k+1]$  where,
	\begin{itemize}
		\item $\player{i}$ has a string $X_i \in \set{0,1}^n$ for each $i \in [k]$,
		\item $\player{i}$ for $1 < i \leq k+1$ has an index $\ind_{i-1} \in [n]$, and a string $X_{i-1}(<\ind_{i-1})$, 
	\end{itemize}
	such that,
	\[
	X_i(\ind_{i}) = z,
	\]
	for some bit $z \in \set{0,1}$. 
	The players have a blackboard visible to all the parties. For $i \in [k]$ in ascending order, $\player{i}$ sends a single message $M_i$, after which the index $\ind_i$, and the string $X_{i-1}(<\ind_{i-1})$ are revealed to the blackboard at no cost. $\player{k+1}$ has to output whether $z$ is 0 or 1. 
	Refer to \Cref{fig:board-aug-chain} for an illustration. 
\end{definition}

\begin{figure}[h!]
	\centering
	\tikzset{every picture/.style={line width=0.75pt}} %set default line width to 0.75pt        

\begin{tikzpicture}[x=0.75pt,y=0.75pt,yscale=-1,xscale=1]
	%uncomment if require: \path (0,409); %set diagram left start at 0, and has height of 409
	
	%Shape: Rectangle [id:dp08092830670644147] 
	\draw   (194,83.27) -- (275.55,83.27) -- (275.55,167) -- (194,167) -- cycle ;
	%Shape: Rectangle [id:dp41515892960218004] 
	\draw  [fill={rgb, 255:red, 155; green, 155; blue, 155 }  ,fill opacity=0.23 ] (111.55,217) -- (504.55,217) -- (504.55,306.27) -- (111.55,306.27) -- cycle ;
	%Straight Lines [id:da24581065779262623] 
	\draw    (238.55,174.27) -- (248.08,213.33) ;
	\draw [shift={(248.55,215.27)}, rotate = 256.29] [color={rgb, 255:red, 0; green, 0; blue, 0 }  ][line width=0.75]    (10.93,-3.29) .. controls (6.95,-1.4) and (3.31,-0.3) .. (0,0) .. controls (3.31,0.3) and (6.95,1.4) .. (10.93,3.29)   ;
	%Straight Lines [id:da6946691856841916] 
	\draw  [dash pattern={on 4.5pt off 4.5pt}]  (214,216) -- (230.8,174.13) ;
	\draw [shift={(231.55,172.27)}, rotate = 111.87] [color={rgb, 255:red, 0; green, 0; blue, 0 }  ][line width=0.75]    (10.93,-3.29) .. controls (6.95,-1.4) and (3.31,-0.3) .. (0,0) .. controls (3.31,0.3) and (6.95,1.4) .. (10.93,3.29)   ;
	%Straight Lines [id:da63027432333877] 
	\draw    (344,173) -- (355.02,213.34) ;
	\draw [shift={(355.55,215.27)}, rotate = 254.72] [color={rgb, 255:red, 0; green, 0; blue, 0 }  ][line width=0.75]    (10.93,-3.29) .. controls (6.95,-1.4) and (3.31,-0.3) .. (0,0) .. controls (3.31,0.3) and (6.95,1.4) .. (10.93,3.29)   ;
	%Straight Lines [id:da5771761551980594] 
	\draw  [dash pattern={on 4.5pt off 4.5pt}]  (317,215) -- (337.61,176.04) ;
	\draw [shift={(338.55,174.27)}, rotate = 117.88] [color={rgb, 255:red, 0; green, 0; blue, 0 }  ][line width=0.75]    (10.93,-3.29) .. controls (6.95,-1.4) and (3.31,-0.3) .. (0,0) .. controls (3.31,0.3) and (6.95,1.4) .. (10.93,3.29)   ;
	%Straight Lines [id:da13905493351938958] 
	\draw  [dash pattern={on 4.5pt off 4.5pt}]  (460,215) -- (484.97,173.71) ;
	\draw [shift={(486,172)}, rotate = 121.16] [color={rgb, 255:red, 0; green, 0; blue, 0 }  ][line width=0.75]    (10.93,-3.29) .. controls (6.95,-1.4) and (3.31,-0.3) .. (0,0) .. controls (3.31,0.3) and (6.95,1.4) .. (10.93,3.29)   ;
	%Straight Lines [id:da4972386821425008] 
	\draw    (125,174) -- (140.34,207.91) -- (143.18,214.18) ;
	\draw [shift={(144,216)}, rotate = 245.66] [color={rgb, 255:red, 0; green, 0; blue, 0 }  ][line width=0.75]    (10.93,-3.29) .. controls (6.95,-1.4) and (3.31,-0.3) .. (0,0) .. controls (3.31,0.3) and (6.95,1.4) .. (10.93,3.29)   ;
	%Straight Lines [id:da05065659828342528] 
	\draw    (503,149) -- (528.55,149.25) ;
	\draw [shift={(530.55,149.27)}, rotate = 180.57] [color={rgb, 255:red, 0; green, 0; blue, 0 }  ][line width=0.75]    (10.93,-3.29) .. controls (6.95,-1.4) and (3.31,-0.3) .. (0,0) .. controls (3.31,0.3) and (6.95,1.4) .. (10.93,3.29)   ;
	%Shape: Rectangle [id:dp8890436433638991] 
	\draw   (92,84.27) -- (173.55,84.27) -- (173.55,168) -- (92,168) -- cycle ;
	%Shape: Rectangle [id:dp774661449368538] 
	\draw   (293,82.27) -- (374.55,82.27) -- (374.55,166) -- (293,166) -- cycle ;
	%Shape: Rectangle [id:dp5126747250137238] 
	\draw   (421,83.27) -- (502.55,83.27) -- (502.55,167) -- (421,167) -- cycle ;
	
	% Text Node
	\draw (125,118) node [anchor=north west][inner sep=0.75pt]   [align=left] {$\displaystyle X_{1}$};
	% Text Node
	\draw (199,92) node [anchor=north west][inner sep=0.75pt]   [align=left] {$\displaystyle  \begin{array}{{>{\displaystyle}l}}
			X_{2}\\
			\sigma _{1}\\
			X_{1}( < \sigma _{1})
		\end{array}$};
	% Text Node
	\draw (129,58) node [anchor=north west][inner sep=0.75pt]   [align=left] {$\displaystyle P_{1}$};
	% Text Node
	\draw (219,56) node [anchor=north west][inner sep=0.75pt]   [align=left] {$\displaystyle P_{2}$};
	% Text Node
	\draw (273,309) node [anchor=north west][inner sep=0.75pt]   [align=left] {Board};
	% Text Node
	\draw (128,229) node [anchor=north west][inner sep=0.75pt]   [align=left] {$\displaystyle M_{1}$};
	% Text Node
	\draw (159,251) node [anchor=north west][inner sep=0.75pt]   [align=left] {$\displaystyle \sigma _{1}$};
	% Text Node
	\draw (237,229) node [anchor=north west][inner sep=0.75pt]   [align=left] {$\displaystyle M_{2}$};
	% Text Node
	\draw (253,255) node [anchor=north west][inner sep=0.75pt]   [align=left] {$\displaystyle \sigma _{2}$};
	% Text Node
	\draw (350,228) node [anchor=north west][inner sep=0.75pt]   [align=left] {$\displaystyle M_{3}$};
	% Text Node
	\draw (397,228) node [anchor=north west][inner sep=0.75pt]   [align=left] {$\displaystyle M_{k}$};
	% Text Node
	\draw (360,256) node [anchor=north west][inner sep=0.75pt]   [align=left] {$\displaystyle .....$};
	% Text Node
	\draw (419,252) node [anchor=north west][inner sep=0.75pt]   [align=left] {$\displaystyle \sigma _{k}$};
	% Text Node
	\draw (533,125) node [anchor=north west][inner sep=0.75pt]   [align=left] {Output $\displaystyle z$};
	% Text Node
	\draw (300,92) node [anchor=north west][inner sep=0.75pt]   [align=left] {$\displaystyle  \begin{array}{{>{\displaystyle}l}}
			X_{3}\\
			\sigma _{2}\\
			X_{2}( < \sigma _{2})
		\end{array}$};
	% Text Node
	\draw (336,59) node [anchor=north west][inner sep=0.75pt]   [align=left] {$\displaystyle P_{3}$};
	% Text Node
	\draw (429,113) node [anchor=north west][inner sep=0.75pt]   [align=left] {$\displaystyle  \begin{array}{{>{\displaystyle}l}}
			\sigma _{k}\\
			X_{k}( < \sigma _{k})
		\end{array}$};
	% Text Node
	\draw (446,56) node [anchor=north west][inner sep=0.75pt]   [align=left] {$\displaystyle P_{k}$};
	% Text Node
	\draw (386,125) node [anchor=north west][inner sep=0.75pt]   [align=left] {$\displaystyle ......$};
	% Text Node
	\draw (157,275) node [anchor=north west][inner sep=0.75pt]   [align=left] {$\displaystyle X_{1}( < \sigma _{1})$};
	% Text Node
	\draw (252,276) node [anchor=north west][inner sep=0.75pt]   [align=left] {$\displaystyle X_{2}( < \sigma _{2})$};
	% Text Node
	\draw (419,274) node [anchor=north west][inner sep=0.75pt]   [align=left] {$ $};
	% Text Node
	\draw (420,278) node [anchor=north west][inner sep=0.75pt]   [align=left] {$\displaystyle X_{k}( < \sigma _{k})$};

\end{tikzpicture}\caption{An illustration of the $\augchain_{n,k}$ problem from \Cref{def:aug-chain}. The solid arrows illustrate that player $\player{i}$ writes a message $M_i$ to the board. The dashed arrows indicate that $\player{i}$ can read the contents of the board. The order in which the messages are sent and indices, strings are released is also shown. } \label{fig:board-aug-chain}
\end{figure}

For the chained version of Augmented Index, the lower bound that at least one player sends $\Omega(n/k^2)$ bits holds true, and the proof follows with minimal changes.  Using this, \cite{DarkDK23} proved lower bounds for interval independent set selection in turnstile streams with weighted intervals. 

We prove the following result about $\augchain_{n,k}$. 

\begin{theorem}\label{thm:aug-chain-lb}
	For any $n, k \geq 1$, any protocol for $\augchain_{n,k}$ with probability of success at least 2/3 requires communication $\Omega(n-k\log n)$. 
\end{theorem} 

We only give a proof sketch detailing the changes needed to prove \Cref{thm:aug-chain-lb}, as most parts of the proof are similar to the proof of  \Cref{thm:main}.

\begin{proof}[Proof Sketch of \Cref{thm:aug-chain-lb}]
	
	Let $\pi$ be any deterministic protocol for Augmented Index which sends at most $s$ bits. 
	
	We begin by proving the analog of \Cref{lem:bias-ind} for biased augmented index, where in addition to having $\rrho$, Bob also has $\rY(< \rrho)$ and has to output $\rY(\rrho)$ after a message from Alice. 
	
	\begin{claim}\label{clm:bias-aug-ind}
		\[
		\en{\rW \mid \rMindex, \rrho, \rY(< \rrho)} \geq H_2 (1/2 + \theta) - \frac2n \cdot (\en{\rMindex}  + \log n).
		\]
	\end{claim}
	
	\begin{proof}
		It is sufficient to prove that 
		\[
		\en{\rW \mid \rMindex, \rrho, \rY(< \rrho)}  \geq \frac1b \cdot \en{\rY \mid \rMindex, \rT},
		\]
		as the rest of the proof follows similar to that of \Cref{lem:bias-ind}. 
		\begin{align*}
			\en{\rW \mid \rMindex, \rrho, \rY(< \rrho)} &\geq \en{\rW \mid \rMindex, \rrho, \rY(< \rrho). \rT}  \tag{by \Cref{fact:en-mi-facts}-(\ref{part:cond-reduce-entropy})} \\
			&= \Exp_{\rT = T} \big[\frac1{b} \cdot \sum_{\rho \in T} \en{\rY(\rho) \mid \rMindex, \rrho = \rho, \rY(<\rho), \rT = T}\big] \tag{as $\rrho$ is uniform over $T$, and $\rW = \rY(\rrho)$} \\
			&=  \Exp_{\rT = T}  \big[\frac1{b} \cdot \sum_{\rho \in T} \en{\rY(\rho) \mid \rMindex, \rY(<\rho), \rT = T}\big] \tag{as $\rY, \rMindex \perp (\rrho = \rho) \mid \rT= T$, by \Cref{obs:ind-cond-M-T}}.
		\end{align*}
		Let us bound the term for each set $T$ separately. We let $T = \set{\rho_1, \rho_2, \ldots, \rho_b} $ where $\rho_i < \rho_j$ if $i < j$ for all $i, j \in [b]$. We know that $Y([n] \setminus T)$ is fixed to be 0. Then,
		\begin{align*}
			\sum_{\rho \in T} \en{\rY(\rho) \mid \rMindex, \rY(<\rho), \rT = T}
			&= \sum_{i \in [b]} \en{\rY(\rho_i) \mid \rMindex, \rY(< \rho_i), \rT = T} \\
			&=  \sum_{i \in [b]} \en{\rY(\rho_i) \mid \rMindex, \rY(\set{\rho_1, \rho_2 ,\ldots, \rho_{i-1}}), \rT = T} \tag{as $Y([n] \setminus T)$ is fixed to be 0, we remove these bits from conditioning} \\
			&=  \en{\rY(T) \mid \rMindex, \rT = T} \tag{by chain rule of entropy \Cref{fact:en-mi-facts}-(\ref{part:chain-rule-entropy})} \\
			&=  \en{\rY \mid \rMindex, \rT = T}. \tag{again as $Y([n] \setminus T)$ is fixed to be 0, we can add these bits }
		\end{align*}
		Thus, we get the required lower bound on $\en{\rW \mid \rMindex, \rrho, \rY(< \rrho)}$ as follows:
		\begin{align*}
			\en{\rW \mid \rMindex, \rrho, \rY(< \rrho)} &\geq\Exp_{\rT = T}  \big[\frac1{b} \cdot \sum_{\rho \in T} \en{\rY(\rho) \mid \rMindex, \rY(<\rho), \rT = T}\big]  \\
			&= \Exp_{\rT = T}  \big[\frac1{b} \cdot \en{\rY \mid \rMindex, \rT = T} \big] \\
			&= \frac1b \cdot \en{\rY \mid \rMindex, \rT}. \qedhere
		\end{align*}
	\end{proof}
	
	The rest of the proof of \Cref{thm:aug-chain-lb} can be obtained by adding random variable $\rX_i(< \rvind_i)$ also to the tuple $\Prot_{<i}$ and following the steps in the proof of \Cref{thm:main}.
\end{proof}
\section{Applications to Streaming}\label{sec:applications}

In this section we give  applications of our main result to independent sets in vertex arrival streams and streaming submodular maximization. 
%hello

\subsection{Independent Sets}

In edge arrival streams, for any graph $G = (V, E)$, the vertex set $V$ with $n$ vertices is given, and the edges $E$ arrive in any arbitrary order. We are required to process the graph in limited space. 

In vertex arrival streams, for any graph $G = (V, E)$, the edges are grouped by their incident vertices. Vertices from $V$ arrive one by one (in any arbitrary order), and when a vertex arrives, all the edges connecting it to any previously arrived vertices are revealed. This makes the vertex arrival stream a strictly easier model than the edge arrival stream, as the order of edges is restricted. 

Indeed, for the maximal independent set problem, we know that finding algorithms in vertex arrival streams is easier; the greedy algorithm produces a maximal independent set in $\Ot(n)$ space, whereas, in edge arrival streams, any algorithm which finds a maximal independent set requires $\Omega(n^2)$ space \cite{AssadiCK19b,CormodeDK19}.

Maximum independent set (MIS), however, is provably hard in both vertex arrival streams and edge arrival streams. It is known that, any algorithm which performs an $\alpha$-approximation of MIS in edge arrival streams requires $\Omega(n^2/\alpha^2)$ space from \cite{HalldorssonSSW12}. In vertex arrival streams, \cite{CormodeDK19} proved a lower bound of $\Omega(n^2/\alpha^7)$. 
They also gave the following connection between $\chain$ problem and MIS in the proof of Theorem 9 of their paper.

\begin{proposition}[Rephrased from \cite{CormodeDK19}]\label{prop:chain-ind-sets}
	For any $\alpha \geq 1$, any algorithm that gives an $\alpha$-approximation of maximum independent sets in vertex arrival streams for $n$-vertex graphs  using space at most $s$ and probability of success at least 2/3 can be used to solve $\chain_{n^2/64\alpha^4, 2\alpha}$ with communication at most $2\alpha \cdot s$ bits and success probability at least $2/3$. 
\end{proposition}

Our lower bound \Cref{thm:main}, along with \Cref{prop:chain-ind-sets} directly gives the following corollary. 

%\worry{I'm going to add the reduction from chain here to add something to this section}
\begin{corollary}\label{cor:ind-sets}
	For $\alpha\geq1$, any $\alpha$-approximation of maximum independent sets in $n$-vertex graphs in vertex arrival streams uses $\Omega(n^2/\alpha^5 - \log n)$ space.  
\end{corollary}

This further reduces the gap between the lower bounds for $\alpha$-approximation of MIS in vertex arrival streams and edge arrival streams by an $\alpha^2$ factor. 

\subsection{Submodular Maximization}

In this subsection, we will summarize our slight improvements to lower bounds for streaming submodular maximization. 

A function $f$ over ground set $V$  from $f: 2^{V} \rightarrow \IR$ is submodular if and only if, for any two sets $A \subset B \subset V$ and for any element $x \in V \setminus B$,
\[
	f(B \cup \{x\}) - f(B) \leq f(A \cup \{x\})-f(A).
\]
This captures the diminishing returns property of any submodular function. 

We are interested in maximization of a monotone submodular function subject to a cardinality constraint. That is, for a given $\ell$, we want to find a subset $S \subset V$ with $\card{S} \leq \ell$ such that for any other set $T \subset V$ with $\card{T} \leq \ell$, $f(S) \geq f(T)$. 

We are given oracle access to function $f$, however, we do not have access to the entirety of the ground set. The elements of the ground set $V$ arrive one by one, and the algorithm has space $s$ to either store the incoming element or to discard it. The algorithm can query the oracle to $f$ with any subset of the elements currently in storage. We want the storage to be roughly the same as the output size, which is $\ell$.

In this model, \cite{KazemiMZLK19} gave an algorithm which finds a $(1/2-\epsilon)$-approximation in $O(\ell/\epsilon)$ space. \cite{FeldmanNSZ20} showed that a better approximation was not possible. They proved that any algorithm which gets a $(1/2+\epsilon)$-approximation uses $\Omega(\epsilon\card{V}/\ell^3)$. They give the following connection to the $\chain$ problem in Theorem 1.3 and Theorem 1.4 of their paper.

\begin{proposition}[Rephrased from \cite{FeldmanNSZ20}]\label{prop:chain-submod-maxim}
	For any $\epsilon > 0$, there exists a constant $\ell_0$ such that for any $\ell \geq \ell_0$, any randomized streaming algorithm which maximizes a monotone submodular function $f:2^{V} \rightarrow \IR$ subject to cardinality constraint of at most $\ell$, using space at most $s$ and with approximation factor at least $(1/2 + \epsilon)$
in expectation	can be used to solve $\chain_{\card{V}/\ell,\ell}$ with probability of success at least 2/3 and communication at most $s \cdot O(\ell/\epsilon)$.
\end{proposition}

We get an improvement  of $\ell$ factor over the current state-of-art lower bound in \cite{FeldmanNSZ20} as a corollary of \Cref{thm:main}.

\begin{corollary}\label{cor:submod-max-chain}
	For any $ \epsilon > 0$, there exists a constant $\ell_0$ such that for any  $\ell \geq \ell_0$, any streaming algorithm that maximizes a monotone submodular function $f: 2^V \rightarrow \IR$ subject to a cardinality constraint of at most $\ell$, with an approximation factor at least $(1/2 + \epsilon)$,  requires $\Omega(\card{V}\epsilon/\ell^2 - \eps \log (\card{V}))$ space.
\end{corollary}

%This further reduces the gap between 

\subsection*{Acknowledgements}

The author is thankful to Sepehr Assadi for introducing them to the problem and for insightful discussions on the proof. The author would also like to thank Parth Mittal for useful comments, and Christian Konrad for introducing them to the Augmented Chain problem.

The author is very grateful to Mi-Ying Huang, Xinyu Mao, Guangxu Yang and Jiapeng Zhang for an illuminating discussion about the problem. They pointed out an important flaw in an earlier version of this work, and the discussion was instrumental for the new proofs in the current version. 

%\clearpage
\bibliographystyle{alpha}
\bibliography{new}
%\clearpage
\appendix

\part*{Appendix}

\section{Protocol for $\indexprob_n$} \label{app:index-ub}

This section gives a protocol for $\indexprob_{n}$ that matches the lower bound of $\Omega(n\delta^2)$ for advantage $\delta \in (0,1/2)$ from \cite{ChakrabartiCKM13}. This protocol is folklore, and we present it here only for completeness.  

\begin{tbox}
	\textbf{Protocol $\pi_{\delta}$ for $\indexprob_n$ with advantage $\Omega(\delta)$ for $\delta \in (0,1/2)$:}
	\begin{enumerate}[label = $(\roman*)$]
		\item Alice and Bob sample a string $ A \in \{0,1\}^n$ and a permutation $P$ of the elements of $[n]$. 
		\item Alice computes string $Y \in \{0,1\}^n$ as follows. For $i \in [n]$,
		\begin{equation}\label{eq:change-index}
		Y(P(i)) = X(i) \oplus A(i).
		\end{equation}
		\item Alice divides  string $Y$ into $n\delta^2$ blocks of size $1/\delta^2$ each. That is,  block $j$ for $j \in [n \delta^2]$ contains bits $Y((j-1)/\delta^2 + 1), Y((j-1)/\delta^2 + 2), \ldots, Y(j/\delta^2)$.
		\item Alice sends the string $Z \in \{0,1\}^{n\delta^2}$ where $Z(j)$ is the majority bit of block $j$ for $j \in [n\delta^2]$. 
		\item Bob finds $j \in [n\delta^2]$ such that for the input index $\sigma$, $\frac{j-1}{\delta^2}<P(\sigma) \leq \frac{j}{\delta^2}$, and outputs $Z(j) \oplus A(\ind)$. 
	\end{enumerate}
\end{tbox}

Let $\rY, \rZ, \rv{P}$ and $\rv{P}(\rvind)$ denote the random variables corresponding to strings $Y, Z$, permutation $P$ and index $P(\ind)$ respectively. 	For any string $A \in \{0,1\}^n$, let $\card{A}_1$ denote the number of 1s in $A$. 

Let us prove that protocol $\pi_{\delta}$ succeeds with the required advantage. We start with a simple observation about the distribution of the input after the transformation in step $(ii)$.

%We need the following observation about the transformation of input $X$ and $\ind$ with public randomness.
\begin{observation}\label{obs:dist-ind-uniform}
	Random variable $\rY$ is uniform over $\{0,1\}^n$ and $\rv{P}(\rvind)$ is uniform over $[n]$.
\end{observation}
\begin{proof}
	The distribution of string $X \oplus A$ is uniform over $\{0,1\}^n$ as the string $A$ is chosen uniformly at random from $\{0,1\}^n$ using public randomness. The permutation $P$ is chosen uniformly at random, and thus, it can map $\ind$ to any value in $[n]$ with equal probability. 
\end{proof}

We also need anti-concentration bounds for the binomial distribution. 

\begin{claim}\label{clm:anti-conc-bin}
For any large $t \in \IN$, and string $B$ sampled uniformly at random from $\{0,1\}^t$, and constant $c \geq 0$, \[
	\Pr\Bracket{\card{\card{B}_1 - \frac{t}2}} < c \sqrt{t}] \leq 2c.
\]
\end{claim}

\begin{proof}
	
	We use \Cref{fact:binom-lower-bound} to get an upper bound on $\binom{t}{t/2}$:
	\begin{align}\label{eq:binom-ub}
		\binom{t}{t/2} \leq 2^{t \cdot H_2(1/2)} \cdot \sqrt{\frac{4t}{2 \pi t^2}} \leq \frac{2^t}{\sqrt{t}}.  
	\end{align}
	
	We know that,
	\begin{align*}
		\Pr\Bracket{\card{\card{B}_1 - \frac{t}2}} < c \sqrt{t}] &\leq \sum_{i= t/2 - c\sqrt{t}}^{t/2 + c\sqrt{t}} \Pr[\card{B}_1 = i] \\
		&= \sum_{i= t/2 - c\sqrt{t}}^{t/2 + c\sqrt{t}} \binom{t}{i} \cdot \frac1{2^t} \\
		&\leq \sum_{i= t/2 - c\sqrt{t}}^{t/2 + c\sqrt{t}} \binom{t}{t/2} \cdot \frac1{2^t} \tag{as $\binom{t}{i} \leq \binom{t}{t/2}$ for all $i \in [t]$} \\
		&\leq \frac1{2^t} \cdot \frac{2^t}{\sqrt{t}} \cdot 2c\sqrt{t} \tag{from \Cref{eq:binom-ub}} \\
		&= 2c. \qedhere
	\end{align*}
\end{proof}

\begin{lemma}\label{lem:proof-index-ub}
	For any $\delta \in [0,1/2]$,  $\pi_{\delta}$ is a protocol for $\indexprob_n$ that communicates $n\delta^2$ bits and succeeds with probability at least $\frac12 + \Omega(\delta)$. 
\end{lemma}
\newcommand{\jstar}{\ensuremath{j^{\star}}}
\newcommand{\cEstar}{\ensuremath{\cE^*}}
\begin{proof}
	The total number of bits sent by Alice is the length of string $Z$ which is $n\delta^2$. 
	Let block $\jstar$ be the block that Bob finds such that the index $P(\ind)$ lies in the block.
	We know that Bob outputs $Z(\jstar) \oplus A(\ind)$, whereas, $X(\ind) = Y(P(\ind)) \oplus A(\ind)$ from \Cref{eq:change-index}. Therefore, Bob outputs the correct answer whenever $Z(\jstar)$ equals $Y(P(\ind))$. 
	We will lower bound this probability. 
	
	In block $\jstar$, we call an index $i$ with $(\jstar-1)/\delta^2 + 1 \leq i \leq \jstar/\delta^2$ a \textbf{good index} if index $i$ contains the majority bit $Z(\jstar)$. We want to lower bound the number of good indices. 
	
	We know the distribution of block $\jstar$ is uniform over $\{0,1\}^{1/\delta^2}$ by \Cref{obs:dist-ind-uniform}.  Thus, we can use \Cref{clm:anti-conc-bin} to infer that the number of 1s in block $\jstar$ lies in the range $\Bracket{\frac1{2\delta^2} - \frac{c}{\delta}, \frac1{2\delta^2} + \frac{c}{\delta}}$ with probability at most $2c$ for constant $c \geq 0$. Let $\cEstar$ denote the event that the number of 1s is \emph{not} lying in this range. This happens with probability at least $1-2c$. 
	
	Conditioned on event $\cEstar$, the number of 1s in block $\jstar$ falls into one of two cases:
\begin{enumerate}[label = $(\roman*)$]
	\item 	At least $\frac1{2\delta^2} + \frac{c}{\delta}$, where the majority bit becomes 1, and the number of good indices is at least $\frac1{2\delta^2} + \frac{c}{\delta}$, or,
	\item At most $\frac1{2\delta^2}-  \frac{c}{\delta}$, and the majority bit becomes 0, and the number of good indices is again at least $\frac1{2\delta^2} + \frac{c}{\delta}$. 
\end{enumerate}	

Therefore, the number of good indices is large whenever $\cEstar$ happens.  Conditioned on $\cEstar$ , and that index $P(\ind)$ lies in block $\jstar$, we get,
	\begin{align*}
	\Pr\Bracket{\textnormal{Index $P(\ind)$ is a good index}} &= \frac{\textnormal{Number of good indices in block $\jstar$}}{\textnormal{Total number of indices in block $\jstar$}} \tag{by \Cref{obs:dist-ind-uniform}, distribution of $P(\ind)$ is uniform over $[n]$} \\
	&\geq \frac1{(1/\delta^2)} \cdot \paren{\frac1{2\delta^2} + \frac{c}{\delta}} \\
	&= \frac12 + c\delta. 
	\end{align*}

When $\cEstar$ does not happen, in the worst case, Bob gets no advantage, and succeeds with probability 1/2.
The probability of success is at least,
\begin{align*}
	\Pr[Z(\jstar) = Y(P(\ind))] &\geq \Pr[\cEstar] \cdot \paren{\frac12 + c\delta} + \Pr[\neg \cEstar] \cdot \frac12 \\
	&= \frac12 \cdot (\Pr[\cEstar] + \Pr[\neg \cEstar]) + c\delta \cdot \Pr[\cEstar] \\
	&= \frac12 + c\delta \cdot \Pr[\cEstar] \\
	&\geq \frac12 + \delta \cdot c(1-2c). 
\end{align*}
The advantage gained by Bob is lower bounded by $\delta \cdot c(1-2c) = \Omega(\delta)$, concluding our proof. 
\end{proof}

\end{document}